\documentclass[11pt]{article}
\usepackage{amsmath}
\usepackage{color}
\usepackage{amssymb}
\usepackage{graphicx}
\usepackage{mathtools}
\usepackage{parskip}
\usepackage{amsthm}

\newtheorem{lemma}{Lemma}

\graphicspath { {Graphics}}
  
\newcommand{\cS}{{\mathcal S}}

\begin{document}

\title{Delineating Parameter Unidentifiabilities in Complex Models}

\author{Dhruva V. Raman, James Anderson, Antonis Papachristodoulou}

\maketitle
\begin{abstract}
Scientists use mathematical modelling as a tool for understanding and predicting the properties of complex physical systems. In highly parameterised models there often exist relationships between parameters over which model predictions are identical, or nearly identical. These are known as structural or practical unidentifiabilities, respectively. They are hard to diagnose and make reliable parameter estimation from data impossible. They furthermore imply the existence of an underlying model simplification. We describe a scalable method for detecting unidentifiabilities, as well as the functional relations defining them, for generic models. This allows for model simplification, and appreciation of which parameters (or functions thereof) cannot be estimated from data. Our algorithm can identify features such as redundant mechanisms and fast timescale subsystems, as well as the regimes in parameter space over which such approximations are valid. We base our algorithm on a novel quantification of regional parametric sensitivity: multiscale sloppiness. Traditionally, the link between parametric sensitivity and the conditioning of the parameter estimation problem is made locally, through the Fisher Information Matrix. This is valid in the regime of infinitesimal measurement uncertainty. 
We demonstrate the duality between multiscale sloppiness and the geometry of confidence regions surrounding parameter estimates made where measurement uncertainty is non-negligible. Further theoretical relationships are provided linking multiscale sloppiness to the Likelihood-ratio test. From this, we show that a local sensitivity analysis (as typically done) is insufficient for determining the reliability of parameter estimation, even with simple (non)linear systems. Our algorithm can provide a tractable alternative. We finally apply our methods to a large-scale, benchmark Systems Biology model of NF-$\kappa$B, uncovering previously unknown unidentifiabilities.
\end{abstract}

\section{Introduction}
The increasing availability of computing power has motivated the mathematical modelling of complex systems in fields as diverse as systems biology, climate science, and economics. 
These models often incorporate many parameters, each  representing an uncertain or varying quantity affecting model output. In such models it is often possible to change individual parameters by an arbitrarily large amount without affecting observed model output, so long as particular relationships between parameters are preserved. The involved parameters are then \emph{structurally unidentifiable}, and cross-sections of parameter space preserving the aforementioned relationships are known as \emph{structural unidentifiabilities}. Determination of such unidentifiabilities is an important, but hard problem in complex models. On the one hand it is a necessary prerequisite to parameter estimation from data: unidentifiable parameters cannot be estimated, and an attempt to do so may result in misleading information. On the other hand, knowledge of unidentifiabilities can provide mechanistic insight. For instance, suppose we know that the observed dynamics of a chemical reaction network model are affected by the product of two rate constants, but not their individual values. Then both rate constants are unidentifiable, while the associated unidentifiability is the product of the rate constants. Knowledge of the latter means that we can rewrite one parameter in terms of the other, both simplifying the model and informing us of a structural system property. 

The determination of structurally unidentifiable parameters in nonlinear differential equation models has been an ongoing research topic for several decades. Multiple algebraic approaches, which directly analyse the system equations, have been suggested, e.g. \cite{Hermann1977, Pohjanpalo1978, Vajda1989b,Ljung1994,Xia2003,Meshkat2009}. Several such approaches rely on the fact that the identifiability problem can be viewed as a special case of the observability problem \cite{Hermann1977}, where parameters are considered as time-invariant states to be estimated. Research on observability in nonlinear systems often makes use of the Observability Rank Condition \cite{Hermann1977}, which analyses the properties (specifically the rank) of a matrix of Lie derivatives of model output over time to determine unobservable model states \cite{Letellier2005, Liu2013, Whalen2015}. Correspondingly, identifiability has been analysed using Differential Geometry-based approaches that exploit parametric symmetries inherent in the aforementioned matrix of Lie derivatives \cite{Vajda1989b, Margaria2001, Evans2002,  Denis-Vidal2000, Denis-Vidal2004, Yates2009}. Many of the approaches mentioned can potentially determine the exact functional form of the unidentifiabilities, but all suffer from a lack of general applicability and scalability, see e.g. \cite{Chis2011}. The aforementioned  issue has motivated the implementation of numerical approaches based on model simulation at multiple points in parameter space, e.g. \cite{Anguelova2007,Sedoglavic2008}. Such approaches can be applied to much larger models, but cannot deal with general nonlinearities, and are not guaranteed to find all unidentifiabilities. 

More recently, data-based approaches to identifiability analysis have been developed. The profile likelihood method \cite{Raue2009} is a scalable numerical method that detects both unidentifiable parameters and associated unidentifiabilities, and has found widespread popularity in the Systems Biology community. It is also able to detect \emph{practical unidentifiabilities} \cite{Vajda1989}, in which large parameter perturbations induce small but nonzero changes in model output. However, it relies on iteratively moving individual parameters, while each time re-optimising model output over the other parameters. As such, a separate analysis is required for each parameter, and only one-dimensional unidentifiabilities are detected. The approach of \cite{Hengl2007} can flag the probable existence of more general unidentifiabilities through extensive model simulation over parameter space. However unidentifiabilities are provided only as probable functionally related groups of parameters; the functional relations themselves are not determined.

There is an underlying duality between the sensitivity of model predictions to parameter perturbation, and the uncertainty associated with parameter estimation in the presence of measurement uncertainty: it is hard to accurately estimate the model parameters from noisy data if a large change in some (combinations of) parameters induces only a small change in model output. Consequently, identifiability is closely related to the concept of model \emph{sloppiness} \cite{Brown2003, Gutenkunst2007, Waterfall2006,Transtrum2010}. Here, the sensitivity of model predictions to infinitesimal parameter perturbation is highly anisotropic, in that the effect of perturbation in sensitive (`stiff') directions exceeds that of insensitive (`sloppy') directions by many orders of magnitude. It has been shown that sloppiness is a common feature in Systems Biology models \cite{Gutenkunst2007}. A sloppy direction suggests the existence of a practical unidentifiability, although the precise correspondence has been debated \cite{Apgar2010, Chis2014}. Conversely, if model sensitivity to parameter perturbation is dominated by a few `stiff' directions in parameter space, it is possible that an underlying, macroscopic model simplification exists \cite{Transtrum2014}. 

 Existing literature \cite{Joshi2006,Vallisneri2008,Raue2009,Hines2014}, has suggested that the local sensitivity characteristics of nonlinear models often poorly approximate their sensitivity to larger perturbations. This suggests that model sloppiness may not be informative in determining the uncertainty associated with parameter estimates when measurement uncertainty is non-infinitesimal. Dually, it implies, when simplifying models, that the most appropriate reduced model is highly dependent on the range of parameters over which validity is intended. This case has been made in e.g.  \cite{Anderson2009,Peixoto2014}. We introduce a  new notion, \emph{multiscale sloppiness}, that quantifies sensitivity anisotropy as a function of the length scale of perturbation considered, relative to a fixed dataset/model prediction. We show how it furthermore relates to the geometry of the set of parameters satisfying a particular Likelihood-ratio hypothesis test. Multiscale sloppiness asymptotically corresponds to the standard formulation of model sloppiness in the limit of decreasing length-scale. We find that the sensitivity characteristics of models can alter drastically as perturbations of increasing magnitude are considered. In addition, both sloppiness and multiscale sloppiness can be highly dependent on the particular parameter vector considered, as we demonstrate subsequently by example. This suggests that caution must be exercised in labelling an entire model structure sloppy, based on analysis at a single parameter vector. Multiscale sloppiness allows for analysis of how the uncertainty region associated with a parameter estimate changes as the signal-to-noise ratio of the data decreases. 

Our formulation of multiscale sloppiness leads to the presentation of a novel, numerical algorithm for unidentifiability detection. We take a particle in parameter space, and allow it to traverse parameter space via the solution of a set of Hamiltonian Equations. These equations are set such that the particle traces over structural unidentifiabilities if they exist, and otherwise over practical unidentifiabilities. Our algorithm has several attractive features. It is highly scalable, as shown in the examples. It does not require optimization of model output over parameter space, which can be computationally expensive and non-convex. Finally, it can detect not only the involved parameters, but also the \emph{functional form} of the unidentifiability. The idea of analysing model structure through evolving particles in parameter space has been previously considered in \cite{Transtrum2014}, where particles evolve along geodesics of the Fisher Information Metric. These trajectories  have different properties to ours: movement along such a curve changes the nominal parameter vector, while always considering infinitesimal magnitude parameter perturbations. Movement along our curves, by contrast, changes the length scale of parameter perturbation considered, relative to a fixed dataset or nominal parameter vector $\theta^*$. More insight on this distinction is provided in Section \ref{ex:exp}. The method of \cite{Transtrum2014} is, moreover, inapplicable when the model under analysis is structurally unidentifiable, and ill-conditioned when it has sloppy parameter vectors: calculation of the geodesic acceleration requires inversion of the Fisher Information Matrix.

\section{Quantifying Variability in Model Output over Parameter Space}
We consider mathematical models incorporating a parameter vector $\theta$, drawn from some set of allowable parameter vectors, $\Theta$, known as the parameter space. Model output, as a function of the parameter vector, is denoted $y(\theta)$, and may be finite or infinite-dimensional. We often make use of a nominal (i.e. reference) output, which either corresponds to experimentally derived data $D$, or model output at a nominal parameter vector $\theta^*$, i.e. $y(\theta^*)$. In order to quantify the degree of output disruption induced by parametric variation, we require a (scalar) cost function on parameter space. This is denoted $C_D(\theta)$ or $C_{\theta^*}(\theta)$, depending on the choice of reference. The lower the value of the cost function, the better that model output at $\theta$ represents the nominal output. The quantity $\nabla^2_{\theta}C_{\theta^*}(\theta^*)$, which is used extensively throughout the paper, will be referred to as the \emph{cost Hessian}.{ We note that the methods of our paper apply only to models for which both a cost function, and its gradient, can be calculated as a function of the parameters. However, our methods apply to any model for which these are calculable. We also assume that cost functions are twice differentiable.}

We now describe a common, statistically motivated choice of cost function used in the parameter estimation problem. As previously stated, the methods of the paper are not restricted to this choice. If we consider measurement noise as a known random variable corrupting model output, then observations of model output are probabilistic. Given noise-corrupted data $D$, we can then take 
\begin{align}
C_D(\theta) = - \log L(\theta;D) + R(D), \label{eq:statDcost}
\end{align}
where $L(\theta;D)$ represents the likelihood of $\theta$, given $D$, and $R(D)$ is some offset term. Note that $L(\theta;D)$ is numerically equal to the probability density of the data $D$, given noise-corrupted model output at parameter vector $\theta$. So parameter vectors that are more likely to produce the data $D$, when corrupted by measurement noise, have a lower cost. In fact, the parameter vector $\theta$ minimising $C_D(\theta)$ is known as the Maximum Likelihood Estimate (MLE).

If, instead of data $D$, we wish to explore model sensitivity around a nominal parameter vector $\theta^*$, we can instead assume that hypothetical data $D(\theta^*)$ is generated as
\begin{align}
D(\theta^*) = y(\theta^*) + \xi, \label{eq:dataGen}
\end{align}
where $\xi$ is a random variable denoting measurement noise. In this case we can take a cost function $C_{\theta^*}(\theta)$ as
\begin{align}
C_{\theta^*}(\theta) =\mathbb{E}_{\theta^*}[- \log L(\theta;D(\theta^*))] - \mathbb{E}_{\theta^*}[\log L(\theta^*;D(\theta^*))].  \label{eq:statTcost}
\end{align}
 This represents the expected value of \eqref{eq:statDcost}, assuming that data $D$ was generated according to \eqref{eq:dataGen} and that $R(D) = - \mathbb{E}_{\theta^*}[\log L(\theta^*;D(\theta^*))]$. This choice of $C_{\theta^*}(\theta)$ is known as the  Kullback-Leibler (KL) Divergence, or relative entropy, between the random variables $D(\theta^*)$ and $D(\theta)$, each generated according to \eqref{eq:dataGen}.
Moreover the cost Hessian $\nabla^2_{\theta}C_{\theta^*}(\theta^*)$ then corresponds to the Fisher Information Matrix at $\theta^*$. 
    
 Let us take a  concrete example to provide context. Suppose that we assume that measurement noise $\xi$ is Gaussian, being distributed as
  \begin{align*}
  \xi \sim \mathcal{N}(0,\Sigma),
  \end{align*}
  for some covariance matrix $\Sigma$. Then the choices \eqref{eq:statDcost} and \eqref{eq:statTcost} of cost function are respectively given by
  \begin{subequations} \label{eq:gaussianCosts}
\begin{align}
 C_D(\theta) &= \langle \big(D - y(\theta)\big), \Sigma^{-1}\big(D - y(\theta) \big)\rangle \label{eq:gaussianCostD} \\
C_{\theta^*}(\theta) &=  \langle y(\theta^*) - y(\theta), \Sigma^{-1}(y(\theta^*) - y(\theta) )\rangle, \label{eq:gaussianCostStar}
\end{align}
\end{subequations}
where $\langle x, y \rangle$ denotes the inner product of vectors $x$ and $y$. So this is the weighted sum of the squared residuals between the nominal output $D$ or $y(\theta^*)$, and the model output $y(\theta)$. 

Minimisers of the cost function represent model parameter vectors that `best approximate' the nominal output. However, there are likely to be many other parameter vectors that are in reasonable agreement with the nominal output. We define
\begin{subequations} \label{eq:epsUncertain}
\begin{align}
&U_D(\epsilon) = \{\theta: C_D(\theta) \leq \epsilon\} \label{eq:epsUncertainD} \\
&U_{\theta^*}(\epsilon) = \{\theta: C_{\theta^*}(\theta) \leq \epsilon\}, \label{eq:epsilonUncertain}
\end{align}
\end{subequations}
as the \emph{$\epsilon$-uncertainty regions}, relative to $D$ and $\theta^*$ respectively. A definition of \eqref{eq:epsUncertain}, specialised to the case of Gaussian measurement noise as considered in equations \eqref{eq:gaussianCosts}, was provided in \cite{Vajda1989}. 
{
When we take the costs \eqref{eq:statDcost} and \eqref{eq:statTcost}, then the $\epsilon$-uncertainty regions gain statistical properties, in both the Bayesian and Frequentist frameworks. We see this as sublevel sets of these cost functions precisely define the set of parameters whose likelihood given the (expected) data exceeds some critical value that is dependent upon the value of $\epsilon$. In Bayesian statistics, an $n$-percent credible region in parameter space is defined as a set of parameters within which random samples of the posterior distribution would fall $n$ times out of $100$. An $\epsilon$-uncertainty region of the form \eqref{eq:epsUncertainD} is then a  \emph{highest posterior density} credible region, under the assumption of a uniform prior. This type of credible region is a minimiser of the volume among the set of credible regions with the same percentage of credibility. In the presence of a more detailed prior distribution, the cost function \eqref{eq:statDcost} can be multiplied through by the prior for the highest posterior density property to hold. Meanwhile, the choice \eqref{eq:statTcost} of cost function represents those parameters whose expected likelihood, given data distributed according to \eqref{eq:dataGen}, would place them within the highest posterior density credible region.

The Frequentist interpretation of $\epsilon$-uncertainty regions is related to the task of hypothesis testing. Suppose, given data $D$, we took a null hypothesis of $D$ being a noise-corrupted observation of $y(\theta^*)$ (and thus distributed according to \eqref{eq:dataGen}). Suppose further that we wanted to test the alternative hypothesis that $D$ was generated from a different parameter vector within the parameter space. A Likelihood-ratio test (LRT) would reject the null hypothesis (i.e. $\theta^*$) if the following condition was satisfied:
\begin{align}
\Lambda_{\theta^*}(D) := \frac{L(\theta^*|D)}{ \sup_{\theta \in \Theta} L(\theta|D)} \leq k^*, \label{eq:lrt}
\end{align}
for some critical value $k^* \in [0,1]$, and where $\Lambda(D)$ is known as the LRT statistic. The significance level of this LRT is given as the $\alpha$ solving
\begin{align}
Pr\left(\Lambda_{\theta^*}(D) \leq k^* | \theta^* \right) = \alpha. \label{eq:epsToSig}
\end{align}
Using the Likelihood-based cost function introduced in \eqref{eq:gaussianCostD}, we can rewrite the LRT statistic as
\begin{align*}
\log \big(\Lambda_{\theta^*}(D) \big) = \inf_{\theta \in \Theta} C_D(\theta) -  C_D(\theta^*).
\end{align*}
Frequentist confidence regions are often generated by inverting the Likelihood-ratio test \cite{Casella2002}. Specifically, one chooses the set of parameters that would have passed the LRT given in \eqref{eq:lrt}, given the data, as the confidence region. This is given by
\begin{align*}
&\left\{ \hat{\theta} \in \Theta: C_D(\hat{\theta}) \leq \inf_{\theta \in\Theta}C_D(\theta) - \log(k^*) \right\} \\
&= U_D(\epsilon), \text{ for } \epsilon =  \inf_{\theta \in\Theta}C_D(\theta) - \log(k^*).
\end{align*}
We see that $U_D(\epsilon)$ is precisely the confidence region gained by inverting a Likelihood ratio test. Moreover the data-independent cost function  \eqref{eq:statTcost} can be written as $C_{\theta^*}(\hat{\theta}) = \mathbb{E}_{\theta^*}[C_D(\hat{\theta})]$, and we have:
\begin{align*}
&\hat{\theta} \in U_{\theta^*}(\epsilon)  \ \ \ \Leftrightarrow \ \ \
 \mathbb{E}_{\theta^*}[\Lambda_{\hat{\theta}}(D)] \leq \exp(\epsilon),
\end{align*}
 where $\Lambda(D)$ is the Likelihood-ratio test statistic defined in \eqref{eq:lrt}. In other words, $U_{\theta^*}(\epsilon)$ precisely defines the set of parameter vectors that would be expected to pass a LRT, given noise-corrupted data generated from $y(\theta^*)$, and a significance level $\alpha$ related to $\epsilon$ by \eqref{eq:epsToSig}. An approximate probability distribution for $\Lambda_{\theta^*}(D)$ can be gained using Wilks' Theorem \cite{Casella2002}. }


 \section{Relating Sloppiness and Unidentifiability}
 In this section we define model sloppiness and unidentifiability, and describe their relationship. Suppose two parameter vectors $\theta^*$ and $\theta$ have identical observed output. In this case, $C_{\theta^*}(\theta) :=0$, and we deem the model \emph{structurally unidentifiable} at $\theta^*$, following \cite{Bellman1970}. The particular parameters that differ between $\theta^*$ and $\theta$ are the structurally unidentifiable parameters. 
An example is the damped harmonic oscillator, modelled by the equation
\begin{align}
m\ddot{y}(t) + c\dot{y}(t) + ky(t) = 0. \label{eq:standardOsc}
\end{align}
Here $y(t)$ denotes vertical displacement as a function of time, and is taken as the observed variable. Meanwhile, $\theta =[m,c,k]$ denotes the mass, damping coefficient, and spring constant respectively, and dotted variables are time derivatives. Although there are three parameters, dividing \eqref{eq:standardOsc} through by $m$ shows that observations are invariant to parameter changes preserving the ratios $\frac{c}{m}$ and $\frac{k}{m}$. Thus, all three parameters are structurally unidentifiable, and we say that there is a structural unidentifiability over areas of parameter space preserving these ratios. As well as providing physical insight, this tells us that the parameter estimation problem is ill-posed for this model: we cannot estimate $\theta$ by observing $y(t)$. A standard reparameterisation of \eqref{eq:standardOsc} is
\begin{align}
\ddot{y}(t) + 2\zeta \omega_0 \dot{y}(t) + \omega_0^2y(t) = 0. \label{eq:reparamOsc}
\end{align}
Here $\omega_0$ and $\zeta$, known as the natural frequency and damping ratio respectively, are given by $\omega_0 = \sqrt{\frac{k}{m}}$, $\zeta = \frac{c}{2\sqrt{mk}}$. The reparameterised model is now identifiable, and the new parameters can be estimated from the observation data.

Even if a model is structurally identifiable at $\theta^*$, which implies that $U_{\theta^*}(0) = \{\theta^*\}$, it may be that the set $U_{\theta^*}(\epsilon)$ remains large for small, but nonzero $\epsilon$. This means that large tracts of parameter space induce very similar model output that cannot easily be distinguished from $\theta^*$ given even slightly noisy data. This is known as practical unidentifiability \cite{Vajda1989}: the effect on parameter estimation is that a small degree of measurement uncertainty results in a large degree of parametric uncertainty. Sloppy models often exhibit the same pathology (\cite{Gutenkunst2007, TranstrumMS2011}), and indeed there is a formal relationship between sloppiness and the properties of $U_{\theta^*}(\epsilon)$, which we present next.

Sloppiness quantifies the anisotropy in output disruption when infinitesimal parameter perturbations are applied to a nominal parameter vector $\theta^*$. Specifically, the degree of sloppiness is the ratio between the degree of disruption in the most and least sensitive directions. Algebraically, this is the condition number of the cost Hessian $\nabla^2C_{\theta^*}(\theta^*)$, i.e. the ratio of its maximal and minimal eigenvalues. Geometrically, this is the aspect ratio of the hyper-ellipse of perturbations $\delta \theta$ satisfying $\langle \delta \theta, [\nabla^2C_{\theta^*}(\theta^*)] \delta \theta \rangle = 1$. The major and minor axes of this hyper-ellipse respectively represent the sloppiest and stiffest directions in parameter space. When the cost Hessian is singular, we will say that the degree of sloppiness is `infinite'. Note, however, that a singular cost Hessian does not imply local structural unidentifiability of a parameter vector, although the converse is true (see the model described in equation \eqref{ex:lpvODE}). Furthermore if the cost Hessian is numerically generated, then it is not possible to discriminate between true singularity and the existence of very small but strictly positive eigenvalues. 
 
  Meanwhile the shape of $U_{\theta^*}(\epsilon)$, for locally structurally identifiable $\theta^*$, is approximated by the same hyper-ellipse (up to scaling) in the limit of decreasing $\epsilon$ (see Appendix \ref{app:section2} and Figure \ref{fig:uncertaintyToEllipse}). This is due to the effect of the second term of the Taylor expansion of $C_{\theta^*}(\theta)$ eventually dominating in \eqref{eq:epsilonUncertain}, as $\epsilon$ decreases. There is a direct statistical interpretation when the cost functions satisfy a condition known as asymptotic normality  \cite{Casella2002} (cost functions satisfying this include \eqref{eq:statDcost},
\eqref{eq:statTcost}, \eqref{eq:gaussianCostD} and \eqref{eq:gaussianCostStar}). This is the famous Cramer-Rao bound, which states that in the limit of decreasing measurement uncertainty, the matrix $[\nabla^2_\theta C_{\theta^*}(\theta)]^{-1}$ and the covariance matrix of the MLE coincide. So sloppiness corresponds approximately to anisotropy in the MLE covariance. However this approximation can fare poorly when measurement uncertainty is non-negligible \cite{Joshi2006,Raue2009}. 
\begin{figure}[h!] 
\centering
\includegraphics[width = 8cm]{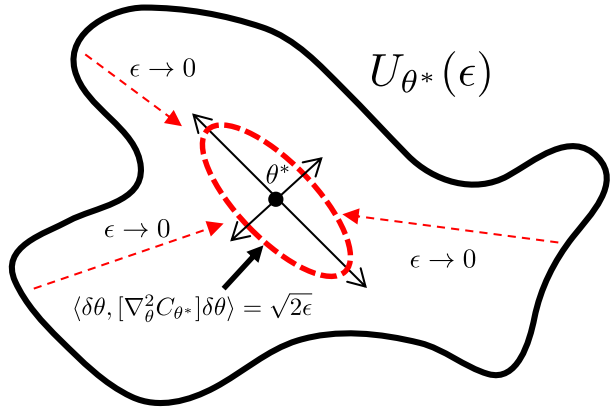}
\caption{\label{fig:uncertaintyToEllipse} As the length-scale $\epsilon$ tends to zero (dashed arrows), the shape of the $\epsilon$-uncertainty region $U_{\theta^*}(\epsilon)$ is increasingly dictated by level sets of $\langle \delta \theta, \nabla^2_{\theta}C_{\theta^*}(\theta^*)\delta \theta \rangle$ (dashed ellipse). This is demonstrated mathematically in Lemma \ref{lem:Ulimit} of Appendix \ref{app:section2}. The sloppiest and stiffest directions are the major and minor axes of the ellipse, respectively.}
 \end{figure}

Sloppiness has been claimed to be a fundamental feature of model structure \cite{Waterfall2006,Gutenkunst2007}, rather than a property of a particular parameter vector. Caution must be exercised in intuiting the former from the latter. Recall the oscillator model \eqref{eq:reparamOsc}, and assume we observe the trajectory, while taking $\theta = [\zeta,\omega_0,y(0),\dot{y}(0)]$, so that initial position and velocity are now both parameters. We take a continuous-time analogue of the squared-residual cost function \eqref{eq:gaussianCostStar}:
\begin{align}
C_{\theta^*}(\theta) = \int^{\infty}_0 \|y(t,\theta) - y(t,\theta^*)\|_2^2 \ dt, \label{eq:l2cost}
\end{align}
which can be reformulated as an algebraic function in the parameters using the methods of \cite{Raman2016}. The degree of sloppiness at the two nominal parameter vectors $\theta^{1,*} = [0.5,1,1,0]$; $\theta^{2,*} = [5,1,1,0]$, are $94.0$ and $3.41 \times 10^7$ respectively, a separation of six orders of magnitude. 

\section{Multiscale Sloppiness}
We now introduce a new notion quantifying the anisotropy of model sensitivity for non-infinitesimal parameter perturbations. This furthermore leads towards an eventual algorithm for uncovering model unidentifiabilities. Define
\begin{subequations}
\begin{align}
&D^{max}_{\bullet}(\delta) = \arg\max_{\theta} C_{\bullet}(\theta) : \|\theta - \theta^*\|^2_2 \leq \delta \label{eq:maxDisrupt} \\
&D^{min}_{\bullet}(\delta) = \arg\min_{\theta} C_{\bullet}(\theta) : \|\theta - \theta^*\|^2_2 = \delta, \label{eq:minDisrupt}
\end{align}
\end{subequations}
where $\|x\|_2$ denotes the Euclidean 2-norm of the vector $x$. Note that \eqref{eq:maxDisrupt} and \eqref{eq:minDisrupt} are set-valued functions: the optimisers may not be unique.
We can think of $D^{max}_{\bullet}(\delta)$ (resp. $ D^{min}_{\bullet}(\delta)$) as the maximally (minimally) disruptive parameter vectors relative to $\theta^*$ at length scale $\delta$. The dual, estimation-based interpretation follows naturally. Let us take $\bar{U}_{\bullet}(\epsilon)$ as the connected component of $U_{\bullet}(\epsilon)$ around $\theta^*$.
Given $\epsilon$, 
we can take $\delta_1$ and $\delta_2$ such that $C_{\bullet}(D^{min}_{\theta^*}(\delta_1)) = \epsilon$, and  $C_{\bullet}(D^{max}_{\theta^*}(\delta_2)) = \epsilon$. Then

\begin{subequations}
\begin{align}
&D^{min}_{\bullet}(\delta_1) = \arg \max_{\theta}  \|\theta - \theta^*\|^2_2: \theta \in \bar{U}_{\bullet}(\epsilon) \label{eq:minDisruptDual}\\
&D^{max}_{\bullet}(\delta^2) = \arg \min_{\theta}  \|\theta - \theta^*\|^2_2: \theta \notin \bar{U}_{\bullet}(\epsilon) \label{eq:maxDisruptDual}.
\end{align}
\end{subequations}
Thus they respectively signify the furthest parameter vectors inside, and the closest parameter vectors outside, the connected component of an $\epsilon$-uncertainty region.  Correspondingly, they represent the furthest (resp. closest) parameter vectors that would be expected to pass (fail) the LRT given in \eqref{eq:lrt} (see Figure \ref{fig:AofDeltavsEps}).

The previous discussion motivates our new definition of \emph{multiscale sloppiness}, at length scale $\delta$, relative to $\theta^*$:
\begin{align}
& \mathcal{S}_{\bullet}(\delta) := \frac{\mathcal{S}^{max}_{\bullet}(\delta)}{\mathcal{S}^{min}_{\bullet}(\delta)} := \frac{\max_{\theta} C_{\bullet}(\theta) : \|\theta - \theta^*\|^2_2 \leq \delta}{\min_\theta  C_{\bullet}(\theta) : \|\theta - \theta^*\|^2_2 = \delta }. \label{eq:newSlop}
\end{align}
\begin{figure}[h!]
\centering
\includegraphics[width = 5cm]{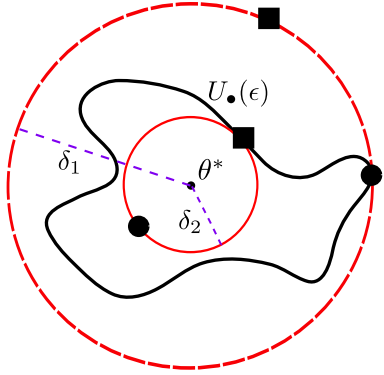}
\caption{\label{fig:AofDeltavsEps} We pick balls of radius $\delta_1$ and $\delta_2$ (dashed) such that  $\mathcal{S}^{min}_{\bullet}(\delta_1) = \mathcal{S}^{max}_{\bullet}(\delta_2)) = \epsilon$. Square and circular dots respectively indicate locations of $D^{max}_{\bullet}(\delta_i)$ and $D^{min}_{\bullet}(\delta_i)$. The balls intersect the uncertainty region $U_{\bullet}(\epsilon)$ at $D^{max}_{\bullet}(\delta_1)$ and $D^{min}_{\bullet}(\delta_2)$.}
 \end{figure}
Note that $\lim_{\delta \to 0} \mathcal{S}_{\theta^*}(\delta)$ is the condition number of the cost Hessian: i.e. the traditional degree of sloppiness. There is no analogue for $ \mathcal{S}_{D}(\delta)$, the data-dependent variant, as the traditional degree of sloppiness cannot account for experimental data. 
We see that $\mathcal{S}_{\theta^*}(\delta)$, for nonzero $\delta$, quantifies anisotropy in the sensitivity characteristics of $\theta^*$ for parameter perturbations of nonzero length-scale. Meanwhile, the numerator and denominator of \eqref{eq:newSlop} are the costs associated with the maximally and minimally disruptive parameter vectors at length-scale $\delta$. Their respective geometrical relationship with the $\epsilon$-uncertainty region has been previously described, and is depicted in Figure \ref{fig:AofDeltavsEps}.

Multiscale sloppiness can be highly dependent on $\delta$. The left-hand side of Figure \ref{fig:overlaidPlots} shows the multiscale sloppiness of the damped oscillator model \eqref{eq:reparamOsc}, with cost function \eqref{eq:l2cost}, at one of the previously considered nominal parameter vectors. In this case, multiscale sloppiness increases with length scale $\delta$. This is not a general rule, however. Consider the model 
\begin{align}
&\begin{bmatrix}
\dot{x}_1(t) \\ \dot{x}_2(t)
\end{bmatrix}
=
  \begin{bmatrix} 
-1 &\phi_1 \\
\phi_1 + \phi_2 &\phi_2 - 4 
\end{bmatrix}
\begin{bmatrix}
{x}_1(t) \\ {x}_2(t)
\end{bmatrix} \ \ \ \ \ \ \ t\geq 0 \nonumber \\
&y(t) = x_1(t) + x_2(t); \label{ex:lpvODE}
 \ \ \ \ \ \ \  \theta = [\phi_1,\phi_2,x_1(0),x_2(0)]
\end{align}
A graph of multiscale sloppiness against length scale for this model is given in Figure \ref{fig:overlaidPlots}, for three nominal parameter vectors and the cost function \eqref{eq:l2cost}. The first parameter vector is structurally identifiable (through the method of \cite{Bellman1970}), yet `infinitely' sloppy, as $\nabla^2_{\theta}C_{\theta^*}(\theta)$ is rank deficient. The second parameter vector is a small perturbation of the first, which has near identical multiscale sloppiness characteristics for nonzero $\delta$, but has nonsingular $\nabla^2_{\theta}C_{\theta^*}(\theta)$. The third parameter vector is much less sloppy at small length scales, but the situation reverses as $\delta$ grows. So if data were generated from $\theta^{*,3}$ with enough measurement noise, the confidence region of the consequent parameter estimate would be expected to be more anisotropic than for $\theta^{*,1}$ and $\theta^{*,2}$, despite its lower level of sloppiness.

 \begin{figure}[h!]
\includegraphics[width = 8.5cm]{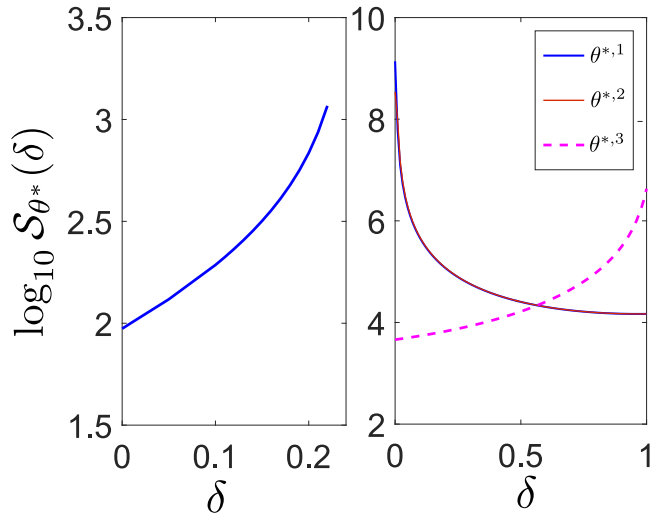}
\centering
\caption{\label{fig:overlaidPlots} Graphs of $\mathcal{S}_{\theta^*}(\delta)$ against $\delta$ for (left) the damped oscillator \eqref{eq:reparamOsc} at $\theta^* = [0.5,1,1,0]$, and (right) \eqref{ex:lpvODE}. The right-hand figure has three choices of nominal parameter vector: $\theta^{*,1} = [0, 0, 1, 2]$, $\theta^{*,2} = [0.001, 0, 1, 2]$, and $\theta^{*,3} = [-0.2, -0.2, 1, 2]$. Not shown on the right-hand graph is that for $\theta^{*,1}$, we have $\mathcal{S}_{\theta^*}(0) = \infty$. }  
 \end{figure}
 
 {
Multiscale sloppiness is formulated with respect to the $L_2$ norm on parameter space (see equations \eqref{eq:minDisruptDual} and \eqref{eq:maxDisruptDual}). Frequently, it is not easy to non-dimensionalise a model so that the parameters are dimensionless quantities. Often, the rescaling $\hat{\theta}_i = \log\theta_i$ is employed so that parametric variation is quantified as a relative, rather than absolute (see e.g. \cite{Gutenkunst2007}), change, in this case. This method is inapplicable when a parameter component with a nominal value of $0$ is considered, and may not be desirable in other circumstances. In such cases, the $L_2$ norm on parameter space depends on the units by which the parameters are measured, and can be weighted by rescaling. Thus the multiscale sloppiness, at a fixed $\delta$, is sensitive to rescaling. However, analysis of the change in multiscale sloppiness properties as the length scale $\delta$ is varied yields model insights that are independent of the parameter units in a dimensionalised model. We now discuss the relationship between multiscale sloppiness, practical unidentifiability, and the scaling of units in parameter space.

 If a modeller is interested in the behaviour of a model to perturbations in a parameter with units of distance, and these perturbations are on the length-scale of millimetres, then the modeller should use units of millimetres for the perturbation. If it turns out that only perturbations on the length-scale of metres have a discernible effect on model output, then the parameter is practically unidentifiable: it cannot be estimated to millimetre precision. Practical unidentifiability is a quantitave phenomenon, dependent on the requirements of the modeller, as opposed to structural unidentifiability, which is a structural phenomenon. In this case, minimally disruptive parameters will align with the co-ordinate axis of the parameter in question.

In general, practical unidentifiabilities cannot be explained as a consequence of poorly chosen units in parameter space, and point to more fundamental structural properties of the model. We provide several examples of such unidentifiabilities subsequently, in Subsection \ref{subsec:nfkb}. We now provide a brief, motivating overview of one such example. Suppose a model contains a fast-timescale subsystem whose dynamics effectively equilibrate on the timescale of the observed dynamics. Increasing the (positive) time-constant $\tau$ of the subsystem will have almost no effect on dynamics. Decreasing the time-constant will also have initially negligible effect, until a critical value $\tau^*$ is reached at which instant equilibration starts to become a poor approximation. Past this point, perturbations to the time-constant will have a significant effect on observed model dynamics. Regardless of the units in which the time-constant is measured, we have one interval of $\tau$ values over which perturbation significantly affects dynamics, and another in which its effect is negligible. The minimally disruptive parameters for length-scale $\delta < \tau - \tau^*$, will be aligned with the co-ordinate axis of $\tau$, and have negligible effect on the dynamics. The multiscale sloppiness will therefore be high. At $\delta \approx \tau - \tau^*$, the multiscale sloppiness will decrease, and the minimally disruptive parameters will diverge from the co-ordinate axis of $\tau$, as further movement in the direction of decreasing $\tau$ becomes increasingly disruptive. Finally, a step change may occur, at which point the minimally disruptive parameters lie in a different direction of parameter space, reflecting a different model approximation.






 }
 
 \section{Minimally Disruptive Curves in Parameter Space}
In the previous section we highlighted the relationship between minimally disruptive parameter vectors and practically unidentifiable subspaces (as quantified through the $\epsilon$-uncertainty region). However, finding the minimally/maximally disruptive parameter vectors for a non-zero length-scale equates to a nonlinear, non-convex optimization which may be very difficult to solve. { Therefore this section instead formulates the concept of minimally disruptive curves in parameter space, which can be traced by evolving a particle satisfying a set of Hamiltonian equations. These curves uncover regions of parameter space over which the cost function stays at, or close, to zero. Indeed these curves are guaranteed to trace over continuous branches of minimally disruptive parameters with strictly monotonically increasing length scale, where the latter exist. Hence, they are also guaranteed to trace over all structural unidentifiabilities with strictly monotonically increasing length-scale. A slight modification of the formulation is further guaranteed to trace over structural unidentifiabilities that do not satisfy the aforementioned strict monotonicity.}

We first define the notion of a `continuous branch' of minimally disruptive parameter vectors. Let us consider any continuous function $ \Xi^\Delta_{min}$ satisfying:
 \begin{align}
& \Xi^\Delta_{min}(\delta) \in D^{min}_{\theta^*}(\delta) \ \forall \delta \in [0,\Delta], \label{eq:choiceD} 
 \end{align}
 where $\Delta$ is the largest considered length-scale.
So such a $\Xi^\Delta_{min}$ exists when there is a continuous curve of minimally disruptive parameter vectors with strictly monotonically increasing distance $\delta$ from $\theta^*$, up to distance $\Delta$. We can pick an analogous function  $\Xi^\Delta_{max}$ with the same properties, but formulated with respect to $D^{max}_{\theta^*}$ and the maximally disruptive parameters. Note that
\begin{align*}
&\frac{\mathrm{d}}{\mathrm{d} \delta} \left[ \Xi^\Delta_{min}(\delta) \right]_{\delta = 0} \text{, and }
\frac{\mathrm{d}}{\mathrm{d}\delta} \left[ \Xi^\Delta_{max}(\delta)  \right]_{\delta = 0}
\end{align*}
lie in the eigenspace of the sloppiest and stiffest eigenvalues of the cost Hessian, respectively. 
 
If a function $\Xi^\Delta_{min}(\delta)$ satisfying \eqref{eq:choiceD} exists, then it solves, for some $F$, the following optimization problem (see Appendix \ref{app:section4}):
  \begin{align}
& \gamma^* = \min_{\gamma \in \Gamma^F_{\theta^*}} : \int_{\gamma} C_{\bullet}: \label{eq:ldMin} 
\\ \Gamma^F_{\theta^*} =  
&\left\{  \gamma:[0,F] \to \Theta: \begin{array}{ll} \gamma(0) = \theta^*;  &\frac{\mathrm{d}}{\mathrm{d}s} \|\gamma(s) - \theta^*\|_2 > 0 \\ \|\gamma'(s)\|_2=1
   \end{array}\right\}. \nonumber
\end{align}
 We therefore refer to a trajectory $\gamma(s)$ satisfying \eqref{eq:ldMin} as a \emph{minimally disruptive curve} of length $F$, relative to $\theta^*$. Note that structurally unidentifiable parameter vectors are global minima of the cost function $C_{\theta^*}$. So if there is a curve of structurally unidentifiable parameters in $\Gamma^F_{\theta^*}$, they will be minimally disruptive.
 
  \begin{figure}[h!]
\includegraphics[width = 8.5cm]{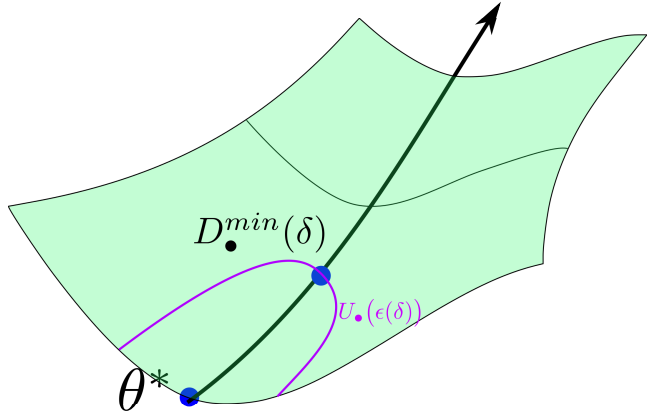}
\centering
\caption{\label{fig:valley}  When a branch of minimally disruptive parameters $D^{min}(\delta)$ vary continuously in $\delta$, their
union is a minimally disruptive curve. Each minimally disruptive parameter on the curve is the (non-unique) furthest point from $\theta^*$ within some (connected) $\epsilon$-uncertainty
region. }  
\end{figure}

A naive approach to solution of \eqref{eq:ldMin} would be a gradient descent style algorithm, where $\gamma'(s)$ points along the direction of steepest descent of $C_{\bullet}\big(\gamma(s)\big)$, subject to appropriate constraints. However in a structurally unidentifiable model, $C_{\theta^*}(\gamma(s))$ is always a local minimum of $C_{\theta^*}$ with null Jacobian $\nabla_{\theta}C_{\theta^*}$. Since exact nullity is never realised in a numerical context, a steepest descent direction will always exist and be highly sensitive to perturbation around the local minima, making numerical evolution ill-conditioned (see Figure \ref{fig:purpleFig}). Instead we can consider \eqref{eq:ldMin} as a constrained variational problem (details in Appendix \ref{app:section3}), and obtain necessary conditions on $\gamma(s)$ to satisfy \eqref{eq:ldMin} by application of Pontryagin's Minimum Principle. These are listed below (and derived in Appendix \ref{app:section3}). Note that the derivative (in $s$) of a function $f$ is denoted $f'(s)$, and the conditions must hold for all $s \in [0,F]$.
 \begin{subequations}
\begin{align}
&C_{\bullet}(\gamma(s)) + \langle \lambda(s),\gamma'(s) \rangle - C_{\bullet}(\gamma(F)) = 0 \\
&  \mu_1(s) (\gamma(s) - \theta^*) + \Big(C_{\bullet}(\gamma(F)) - C_{\bullet}(\gamma(s))\Big) \gamma'(s)= \lambda(s),
\end{align}
where
\begin{align*}
&\lambda(0) = \Big(C_{\bullet}(\gamma(F)) - C_{\bullet}(\gamma(0))\Big)\gamma'(0) \\
& \lambda'(s) = \mu_1(s)\gamma'(s) - \nabla_{\theta}C_{\bullet}(\theta) &\mu_1(s) \geq 0.
\end{align*}
\label{eq:actualHam}
\end{subequations}
 Here $\lambda(s)$ is known as the costate, and is directly analogous to the momentum of Hamiltonian Mechanics. For a given parameter vector $\gamma(s)$, we assume that the cost function and its gradient, i.e. $C_{\bullet}(\gamma(s))$ and $\nabla_{\theta} C_{\bullet}(\gamma(s))$, can be calculated. If we additionally know $\lambda(s)$ and $ C_{\bullet}(\gamma(F))$, then the remaining unknowns $\mu(s)$ and $\gamma'(s)$ are linear in 
\eqref{eq:actualHam}, and can be solved for. Thus $\gamma(s)$ and $\lambda(s)$ uniquely evolve as a system of coupled differential equations in $s$, given a choice of $\gamma'(0)$ and $C_{\bullet}(\gamma(F))$. For \eqref{eq:ldMin} to hold as $s$ approaches zero, $\gamma'(0)$ must be within the eigenspace of the smallest eigenvalue of the cost Hessian. Its magnitude is also fixed. Thus there are $2e$ choices of $\gamma'(0)$,  where $e$ is the dimension of the aforementioned eigenspace. { Note that for $e>1$, the $2e$ choices should form an orthornormal basis in this eigenspace, but such a basis is not unique. This situation occurs in the example of Subsection \ref{ex:jakStat}, which contains a two-dimensional structurally unidentifiable subspace. The discussion accompanying this example provides insight on the appropriate choice of basis.}

The only free variable is now $C_{\bullet}(\gamma(F))$, which sets the initial momentum of the particle tracing the curve, and specifies the cost-cutoff at which the trajectory terminates. Note that this initial momentum ensures that we avoid the previously described pitfalls of a steepest descent approach to solving \eqref{eq:ldMin} (see Fig. \ref{fig:purpleFig}).

 A maximally disruptive curve can be defined by turning the minimisation of \eqref{eq:ldMin} into a maximisation. Such a curve maintains an analogous relationship with  $\Xi^\Delta_{max}$, if it exists. However it is of less interest from the point of view of unidentifiability characterisation. Note that in this case, a gradient-ascent style approach is sufficient from a computational viewpoint.

\begin{figure}[h] 
\includegraphics[width = 8cm]{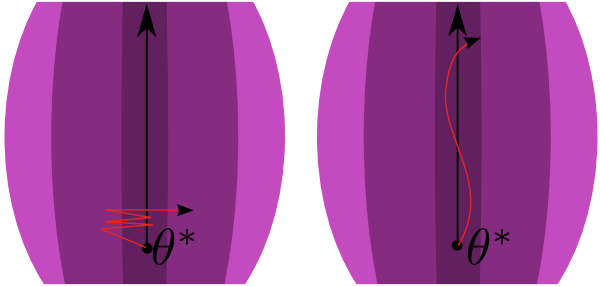}
\centering
\caption{\label{fig:purpleFig} Bird's eye view of level sets of the cost function. Darker areas have lower cost, while the black arrow traces the root of the cost function (i.e. a structural unidentifiability). Red lines denote the trajectory of a particle attempting to evolve numerically along the structural unidentifiability. On the left, a steepest descent style algorithm fails due to the stiffness of the direction of descent around a local minimum. On the right: the inherent momentum associated with trajectory evolution via solution of Hamiltonian equations ensures that numerical error does not completely disrupt evolution.}
 \end{figure}

{ Note that model analysis using minimally disruptive curves is not restricted to models for which a continuous branch of minimally disruptive parameters (as given in \eqref{eq:ldMin}) exists. Instead, existence of the function given in \eqref{eq:ldMin}  implies that the graph of minimally disruptive parameters is a minimally disruptive curve. An unidentifiability may not have strictly monotonically increasing distance from $\theta^*$, making it untraceable by a minimally disruptive curve. In the case of structural unidentifiabilities, we can easily modify the optimization \eqref{eq:ldMin} to accommodate this. We would replace the condition $\frac{\mathrm{d}}{\mathrm{d}s} \|\gamma(s) - \theta^*\|_2 > 0 $ with 
\begin{align}
\frac{\mathrm{d}}{\mathrm{d}s} \|\gamma(s) - \theta^*\|_2 - c \|\gamma(s) - \theta^*\|_2 > 0, \label{eq:monotoneRelax}
\end{align}
for some constant $c \in [0,2)$. This would trivially modify the Hamiltonian equations associated with the solution of \eqref{eq:ldMin}, and derived in Appendix \ref{app:section3}. In this case, the optimization would be guaranteed to trace over any structurally unidentifiable curve $\gamma(s)$, whose angle $\phi(s)$ between $\gamma'(s)$ and the vector $\theta^* - \gamma(s)$ was always greater than $\cos^{-1}\big(\frac{c}{2}\big)$. This guarantee is a result of global minimality of every structurally unidentifiable parameter vector, with respect to the cost function.

In the case of practical unidentifiabilities, the approach proposed in the previous paragraph is not guaranteed to be effective, although it pragmatically recovers highly practically unidentifiable subspaces not obeying monotonicity of distance from $\theta^*$. For less sharply defined practically unidentifiable subspaces, usage of the condition \eqref{eq:monotoneRelax} with $c>0$ can result in a curve that confines itself to a ball around $\theta^*$, outside of which all parameters have non-negligible cost. Thus, the curve will not be informative with respect to parameter perturbations with a larger length-scale than the ball's radius. This can be heuristically tackled by increasing the initial momentum imparted to the particle tracing the curve.

The formalism of minimally disruptive curves uncovers unidentifiabilities of a continuous nature. In this way, functional relations on parameter space over which model output changes little may be recovered. Models may also possess isolated parameter vectors that induce the same, or similar, output (see e.g. the example of Subsection \ref{ex:exp}). If so, the cost function will possess multiple isolated local minima. The $\epsilon$-uncertainty regions may then consist of multiple, unconnected regions each surrounding a local minimum. 
Finding all local minima corresponds to a global optimization problem on the cost function, for which even evaluation often requires a numerical model simulation. This is in general an NP-hard problem \cite{Moles2003}. While local minima of the cost function may well correspond to minimally disruptive parameters at the relevant length-scale, they are not recoverable through the generation of minimally disruptive curves.


}


\section{Examples} \label{sec:ex}

We provide three examples of our methods. The first illustrates concepts introduced in previous sections, while highlighting the importance of non-locality in sensitivity quantification. The second verifies our methods on a benchmark model used in identifiability analysis. The third provides a more detailed analysis of a different benchmark example from the identifiability literature, specifically a metabolic reaction network. Novel structural and practical unidentifiabilities are identified. It is shown how our method can be used to gain additional mechanistic insight. Specifically, we uncover timescale-separated subsystems (and the regimes in which the timescale separation is valid),  unnecessary mechanisms, and model approximations (together with their regimes of validity). 

\subsection{Illustrative Example} \label{ex:exp}
We first use a simple example to illustrate the concepts introduced in this paper, and compare them against the existing sloppiness framework. The model, previously considered in \cite{TranstrumMS2011, Transtrum2014}, is given as:
\begin{align}
&y(\theta) = \left[ e^{\frac{-\theta_1}{3}} + e^{\frac{-\theta_2}{3}}, e^{-\theta_1} + e^{-\theta_2}, e^{-3\theta_1} + e^{-3\theta_2} \right]. \label{eq:expModel}\\
&C_{\theta^*}(\theta) = \|y(\theta) - y(\theta^*)\|_2^2 \nonumber 
\end{align}
We take parameter space as $\Theta = [0,4] \times [0, 4]$. We consider a nominal parameter vector $\theta^* = [4,0.5]$. We deliberately place $\theta^*$ on the boundary of $\Theta$ so that only one direction of the sloppiest eigenvector of $\nabla^2_{\theta}C_{\theta^*}(\theta)$ need be considered. This halves the analysis required without sacrificing illustrative power. 

Note that model output \eqref{eq:expModel} corresponds to observing a mixed exponential decay of the form $e^{-\theta_1 t} + e^{-\theta_2 t}$, at the timepoints $t=\{\frac{1}{3},1,3\}$, and so model output is invariant with respect to permutation of the parameters. The model at $\theta^*$ is therefore structurally identifiable in a local neighbourhood but structurally unidentifiable from a global perspective (for nonzero parameter vectors). In such a case, a minimally disruptive curve is not guaranteed to uncover the structurally unidentifiable parameter vector $[0.5, 4]$: the guarantee only applies when structural unidentifiability is local. In this example, by chance, the minimally disruptive curve does in fact trace over $[0.5, 4]$. 

We demonstrate in Figure \ref{fig:expComp} the differences inherent to evolving a particle in parameter space along its sloppiest eigenvector (which retains no memory of the initial vector $\theta^*$), as compared to evolving a minimally disruptive curve (whose evolution is explicitly dependent on $\theta^*$). To these ends, we first construct a vector field that, when evaluated at some $\theta \in \Theta$, points in the direction of the sloppiest eigenvector of $\nabla^2_\theta C_{\theta}(\theta)$. The flow of this vector field, initialised from some $\theta^*$ (which we take to be $[4,0.5]$), is compared to a minimally disruptive curve emanating from the same $\theta^*$. We see that the former, although by definition always pointing in the sloppiest direction (to infinitesimal magnitude), diverges from the latter. Indeed, the flow of the vector field takes a higher cost route, as shown on the figure, tracing over parameter vectors with more highly diverging output (relative to $y(\theta^*)$) than the minimally disruptive curve (as required by definition). This discrepancy arises from the fact that the flow of the vector field, at some $\theta \in \Theta$ is determined independently of $\theta^*$, and by analysis of the matrix $\nabla^2_{\theta}C_{\theta}(\theta)$. The direction minimising local change in $\nabla^2_{\theta}C_{\theta}(\theta)$, is not necessarily the direction minimising non-local change in $C_{\theta^*}(\theta)$. Furthermore the minimal eigenvector of the cost Hessian at the nominal, i.e. $\nabla^2_{\theta^*}C_{\theta}(\theta)$, is a poor marker of the direction of minimal sensitivity to non-infinitesimal parameter perturbation.

\begin{figure}[h] 
\includegraphics[width=8.6cm]{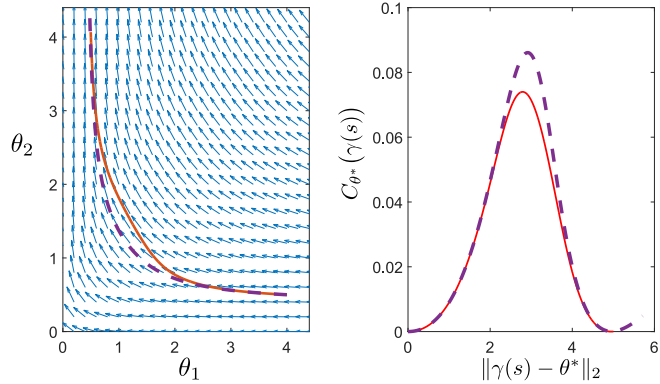}
\centering
\caption{\label{fig:expComp} Left: Vector field formed by sloppiest eigenvectors of the cost Hessian (one direction only), for the model \eqref{eq:expModel}. Overlaid are (solid line, red) a curve of minimally disruptive parameters, and (broken line, purple) the flow of the vector field. Right: The costs associated with these two curves. The x-axis charts Euclidean distance from $\theta^*$, while the y-axis charts the associated cost. }%
 \end{figure}

\subsection{IL13-induced JAK-STAT Pathway} \label{ex:jakStat}
 We next verify our methods on a differential equation model of the IL13-induced JAK-STAT pathway \cite{Raia2011}, which was used as a benchmark for identifiability analysis in \cite{Raue2014}. The model has a $23$-dimensional parameter space (further details in the Supplementary Information) together with published estimates of each parameter, constituting $\theta^*$. In \cite{Raue2014}, two structural unidentifiabilities were found: model output was invariant over sets of the form
\begin{subequations}
\begin{align}
& {\theta^*_{17}\theta^*_{22}} = {\theta_{17}\theta_{22}};  \label{eq:jakSubspace1} \\
& {\theta^*_{15}\theta^*_{21}}={\theta_{15}\theta_{21}} \text{ and }  {\theta^*_{15}\theta_{11}}={\theta^*_{11}\theta_{15}}.\label{eq:jakSubspace2}
\end{align}
\end{subequations}
Only the method of \cite{Raue2009} succeeded in detecting both the parameters involved in \eqref{eq:jakSubspace1}, \eqref{eq:jakSubspace2}, and their functional relations. This method requires evolution of a separate curve for each of the $23$ parameters.  To generate the $i^{th}$ curve, $\theta^*_i$ is varied incrementally, and the other parameters iteratively re-optimised so as to minimise the effect on the cost function. This has the advantage of not requiring computation of $\nabla_{\theta}C_{\bullet}$ (unlike our method), and also generates componentwise confidence intervals. However the re-optimization steps can be expensive.

 Our method recovered both \eqref{eq:jakSubspace1} and \eqref{eq:jakSubspace2}, requiring only four curve generations (evolvable in parallel). The Cost Hessian, at $\theta^*$, had two eigenvectors with numerically zero eigenvalues. Therefore the four curves had their initial velocity $\gamma'(0)$ aligned in both the positive and negative directions of each of these eigenvectors. The trajectories are shown in Figure \ref{fig:jakStat}. The top half of Figure \ref{fig:jakStat}a shows the log-space evolution of each component of $\gamma(s)$, relative to $\theta^*$. The components corresponding to the rate constants $\theta_{17}$ and $\theta_{22}$ are marked, and we see that the summation $\log \theta_{17}  + \log \theta_{22}$ is preserved to numerical precision. The lower half of the figure shows, meanwhile, that the cost $C_{\theta^*}(\gamma(s))$ is zero, to numerical precision. This implies a structural unidentifiability over areas of parameter space preserving the product $\theta_{17}\theta_{22}$, as given in \eqref{eq:jakSubspace1}. Similarly, the trajectory in  Figure \ref{fig:jakStat}b  preserved the expressions $\log \theta_{15} - \log \theta_{11}$, and $\log \theta_{15} + \log \theta_{21}$, implying the unidentifiability of \eqref{eq:jakSubspace2}.

\begin{figure}[h] 
\includegraphics[width=8.6cm]{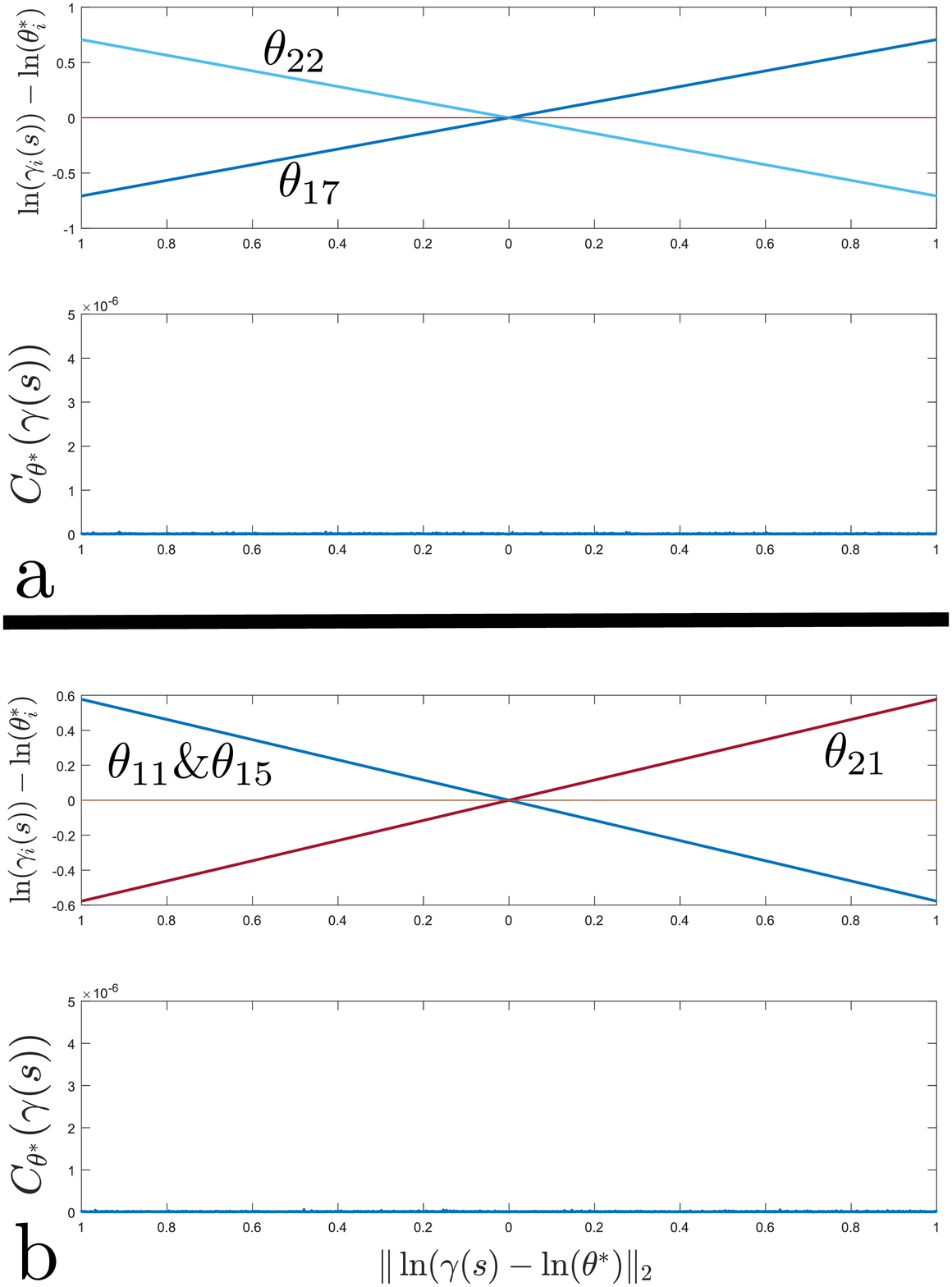}
\centering
\caption{\label{fig:jakStat} The $x$-axes of all plots in the panels depict the distance, in log-space, between the nominal parameter vector $\theta^*$, and that of $\gamma(s)$, as $s$ varies. Each line on the top-half figures of a and b represents the relative change of a particular (labelled) parameter component $\theta_i$, from $\theta^*_i$, as $s$ varies. The bottom half figures represent the cost associated with the parameter vector $\gamma(s)$, relative to $\theta^*$, as $s$ varies. Note that in the top half of figure b, both $\theta_{11}$ and $\theta_{15}$ are represented by the same line: they are overlaid so closely as to be indistinguishable.}%
 \end{figure}

The fact that the JAK-STAT model has two, one-dimensional, structural unidentifiabilities implies that the Cost Hessian has a two-dimensional null-space. As stated, two eigenvectors spanning the null-space set the initial directions $\gamma'(0)$ of the four generated curves. It is possible, however, to synthesise multiple such eigenvector pairs, as a two dimensional linear space does not have a unique orthonormal basis. Using the premise that different structural unidentifiabilities will involve different parameters, we used an eigenvector pair whose nonzero entries did not correspond at all, resulting in the curves of Figure \ref{fig:jakStat} and the unidentifiabilities of \eqref{eq:jakSubspace1} and \eqref{eq:jakSubspace2}. If we denote the eigenvectors chosen as $v_1$ and $v_2$, then the choice $\tilde{v}_1 = v_1 + \frac{2}{3}v_2$ and $\tilde{v}_2 = v_1 - v_2$ also forms an orthonormal basis. Both, however, have nonzero components for all five parameters involved in the structural unidentifiabilities. Generating minimally disruptive curves whose initial directions aligned with $\tilde{v}_1$ and $\tilde{v}_2$ yielded different structurally unidentifiable curves, which involved all five parameters, but satisfied \eqref{eq:jakSubspace1} and \eqref{eq:jakSubspace2}. The latter relations could be uniquely obtained from the curves by applying basic linear algebra to the linear relations preserved in the logarithms of the parameters.

\subsection{NF-$\kappa$B Regulatory Module} \label{subsec:nfkb}
We next find novel structural and practical unidentifiabilities in a benchmark differential equation model used for identifiablity analysis in the literature: the $29$-parameter model of the NF-$\kappa$B regulatory module \cite{Lipniacki2004}. The model describes an intracellular mechanism mediating dynamics of the NF-$\kappa$B protein complex. A schematic diagram, based on one provided in \cite{Lipniacki2004}, is given in Figure \ref{fig:nfkbNetwork}, and details on the system equations and parameters are provided in the Supplementary Information. The system is excited by extracellular TNF. Previous analyses \cite{Chis2011, Lipniacki2004} declared the model structurally identifiable, under relaxed assumptions necessary to maintain viability of the respective algorithms. Specifically, identifiability held so long as all parameters could be discriminated on the basis of some (unknown) time-course of TNF excitation, rather than the binary TNF signal assumed in \cite{Lipniacki2004}. In \cite{Chis2011}, sixteen parameters were furthermore fixed to maintain computational tractability. We analyse the full model, taking $\theta^*$ as the (complete) set of published parameter estimates/assumptions provided in \cite{Lipniacki2004}.

\begin{figure}[h] 
\includegraphics[width=9cm]{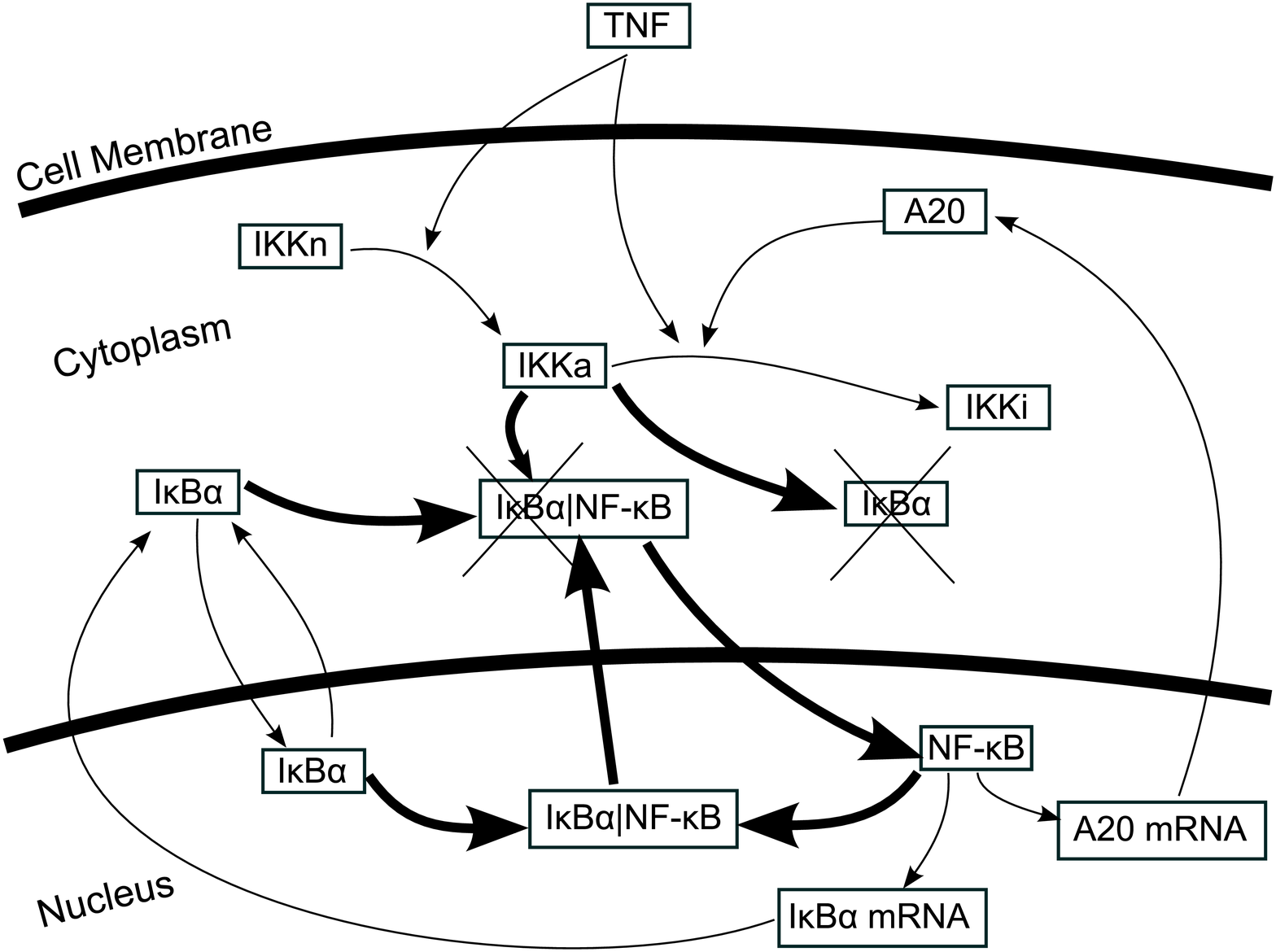}
\centering
\caption{\label{fig:nfkbNetwork} Schematic of the NF-$\kappa$B model from \cite{Lipniacki2004}, based on Figure 2 in the same paper. Bold arrows denote fast-timescale reactions. Crosses denote degradations. Extracellular Tumour Necrosis Factor (TNF) transforms neutral  I$\kappa$B kinase (IKKn) into its active form IKKa. The IKKa catalyses degradation of I$\kappa$Ba, both in its free state, and when bound in the I$\kappa$Ba|NF-$\kappa$B complex. Catalysis of the previous complex frees cytoplasmic NF$\kappa$B, which is then shuttled to the nucleus. Nuclear NF-$\kappa$B promotes both  I$\kappa$Ba and A20 transcripts. The translated A20 inactivates IKKa, thus indirectly promoting cytoplasmic I$\kappa$Ba. The translated I$\kappa$Ba binds to NF-$\kappa$B, making the latter inert and shuttling it to the cytoplasm. So nuclear NF-$\kappa$b is involved in a negative feedback loop with I$\kappa$Ba, and a positive feedback loop with A20. }%
 \end{figure}

We evolved two curves satisfying \eqref{eq:actualHam} in parallel, setting $\gamma'(0)$ in both the positive and negative directions of the eigenvector corresponding to the minimal eigenvalue of the cost Hessian. The result is shown in Figure \ref{fig:collatedSubspaces}a, which connects both curves at $\theta^*$. The top half of the figure shows the log-space evolution of each component of $\gamma(s)$, relative to $\theta^*$. The components corresponding to the rate constants $k_2$ and $c_4$ are marked, and we see that the summation $\log k_2  + \log c_4$ is preserved to numerical precision. The lower half of the figure shows, meanwhile, that the cost $C_{\theta^*}(\gamma(s))$ is zero, to numerical precision. This implies a structural unidentifiability over areas of parameter space preserving the product $k_2c_4$. This makes sense, as $k_2$ and $c_4$ respectively represent the translation (i.e. production) rate of the protein $A20$, and the inactivation rate of the kinase IKKa caused by $A20$. Since $A20$ and its transcript are not directly observed, it is clear that $A20$ only enters the model through its effect on IKKa. This refines the conclusion of \cite{Lipniacki2004}, where the effect of $A20$ on dynamics was proposed to be a combination of its \emph{concentration} (rather than translation rate), and the inactivation rate. 

We next rewrote $k_2$ as $k_2 = \frac{k_2^*c_4^* }{c_4}$, thus reducing the dimensionality of parameter space by one, and eliminating the previous structural unidentifiability. Two minimally disruptive curves were again evolved in parallel, with initial directions in the positive and negative directions of the sloppiest eigenvector. Results are shown in Figure \ref{fig:collatedSubspaces}b. Three parameters show significant movement, and are involved in a practical unidentifiability. They are listed below:
\begin{itemize}
\item $e2a$ represents the rate of transport of the I$\kappa$B$\alpha$|NF-$\kappa$B complex from the cytoplasm to the nucleus. This was appreciated as a fast-timescale reaction in \cite{Lipniacki2004}, as shown in Figure \ref{fig:nfkbNetwork}.
\item $c_{5a}$ is the constitutive degradation rate of free cytoplasmic I$\kappa$B$\alpha$.
\item $t_1$ is the rate constant for catalysis of the  I$\kappa$B$\alpha$-IKKa complex into IKKa, also described as a fast timescale reaction in \cite{Lipniacki2004}.
\end{itemize}
By varying these three parameters in turn, we see that each individual parameter has a close to zero effect on model output when varied, and the unidentifiability does not depend on their correlated change. Since $e_{2a}$ and $t_1$ are fast-timescale reaction rates, the involved reactants and products are equilibrated near-instantaneously at $\theta^*$, and increasing the rates makes no difference to observed model dynamics. Similarly, Figure \ref{fig:nfkbNetwork} implies that  I$\kappa$B$\alpha$-IKKa degradation is dominated by terms other than constitutive degradation, and decreasing $c_{5a}$ will therefore barely affect observed dynamics. Thus, elimination of the reaction corresponding to the rate $c_{5a}$, and instant equilibration of the reactions corresponding to $t_1$ and $c_{5a}$, would negligibly affect dynamics at $\theta^*$. 

Note from Figure  \ref{fig:nfkbNetwork} that as the minimally disruptive curve travels in the direction of decreasing $e_{2a}$, decreasing $t_1$, and increasing $c_{5a}$, the associated cost starts to increase, until a step change in behaviour occurs. This can be explained mechanistically. Decrease in $t_1$ and $e_{2a}$ slows the timescale of their associated reactions, until instant equilibration is no longer a valid approximation and observed dynamics are affected. Meanwhile, as $c_{5a}$ increases it starts to make a meaningful contribution to  I$\kappa$B$\alpha$ degradation, and thereby also affects observed dynamics. Crucially, our algorithm finds at what parameter values these step changes occur, thereby defining the areas of parameter space over which the instant equilibration/zero degradation approximations are valid. 

We fixed the parameters $c_{5a}$, $e_{2a}$, and $t_1$ and again evolved minimally disruptive curves as previously. Results are shown in Figure \ref{fig:collatedSubspaces}c. We see that in a region around $\theta^*$, dynamics are insensitive to parameter changes preserving the product $c_{1a}c_{4a}$. The former is the rate constant for promotion of the mRNA transcript for I$\kappa$B$\alpha$ by free nuclear NF-$\kappa$B. The latter is the translation rate for I$\kappa$B$\alpha$. Since the production rate of a protein is the product of its transcription and translation rates, then, as long as mRNA transcription is dominated by $c_{1a}$, it makes sense that preservation of the product $c_{1a}c_{4a}$ should not affect observed dynamics. However, the figure shows that this unidentifiability breaks down when $c_{1a}$ decreases sufficiently. At this point, the cost starts to increase, and a step change in the practical unidentifiability is observed. The reason is that the model incorporates constitutive transcription of I$\kappa$B$\alpha$, which becomes the dominant contributor to overall transcription in this regime. 

Again, we removed the practical unidentifiability in $c_{1a}c_{4a}$ by substituting $c_{1a} = \frac{c_{1a}^*c_{4a}^*}{c_{4a}}$, and evolved minimally disruptive curves (Figure \ref{fig:collatedSubspaces}d). The resulting curve involved $a_2$, which represents the association rate of IKKa and  I$\kappa$B$\alpha$ in the nucleus. In \cite{Lipniacki2004}, it is depicted as a fast timescale reaction, with a nominal value of $0.2$ obtained by assumption. In fact, as shown in Figure \ref{fig:collatedSubspaces}d, we can decrease $a_2$ by over twenty orders of magnitude without affecting nominal dynamics, so long as other parameters are changed in a compensatory manner. Thus we could switch off this mechanism with only a negligible change in nominal dynamics.

For our final curve evolution, we fixed $a_2$, with results shown in Figure \ref{fig:collatedSubspaces}e. This results in a complicated practical unidentifiability involving many parameters changing by several orders of magnitude without perceptibly affecting dynamics. Iteratively fixing three further parameter vectors whose values changed considerably in Figure \ref{fig:collatedSubspaces}e also resulted in complicated practical unidentifiabilities involving large parameter variations with little effect on the cost (these are not shown graphically).
We could not attach mechanistic meaning to the practical unidentifiability of Figure \ref{fig:collatedSubspaces}e, and indeed it may not exist (see e.g. \cite{Brenner2010} for a more detailed discussion).

  To demonstrate, suppose that we only measured model output at a small subset of timepoints. A whole host of new unidentifiabilities would arise, whose particular form would be dependent on the subset chosen, but which would not reflect mechanistic features of the model. These would be eliminated by measurement at more timepoints. Similarly, extra information from the NF-$\kappa$B model may be required to sufficiently constrain its parameter space. For instance, a knockdown experiment inhibiting expression of a particular gene could be conducted, and the appropriately modified model fitted to the consequent data. We could then require that for a parameter vector to be low-cost, it must not only recreate a set of nominal dynamics, but also behave appropriately when modified to model the knockdown experiment.
 
The modeller may wish to estimate a particular subset of the parameter vector. In this case, further generation of minimally disruptive curves, and consequent fixation of parameters, would be recommended. At some point, minimally disruptive curves would start to induce significant perturbations of the cost function over small length-scales of perturbation. At this point a set of parameters would have been found whose \emph{a priori} fixation would result in well-posed parameter estimation for the remaining parameters of the model.

\section{Discussion} \label{sec:disc}
We have provided a numerical approach to elucidation of both unidentifiable parameters and their associated unidentifiabilities. It can complement symbolic approaches to finding unidentifiable parameters, which additionally provide theoretical guarantees, but can suffer from scalability issues. In particular, algebraic approaches based on the Observability Rank Condition \cite{Hermann1977} are in general computationally intensive even for modestly-sized models, but work well when there are few parameters and many outputs under consideration. Our method can be used as a preconditioning tool for such algorithms, as it can highlight groups of parameters likely to be structurally unidentifiable. All other parameters can then be fixed, and unidentifiability of the simplified model checked algebraically. For instance, in the example of the NF-$\kappa$B network from \cite{Lipniacki2004}, we fixed all parameters except $k_2$ and $c_4$ (which we discovered were structurally unidentifiable through our algorithm, as discussed previously and shown in Figure \ref{fig:collatedSubspaces}a). The reduced, two-parameter model could then be easily analysed algebraically, and verified as structurally unidentifiable by direct application of the Observability Rank Condition. The same, integrated procedure was also performed successfully to verify the structurally unidentifiable parameters of the JAK-STAT model analysed previously. 

Our method additionally inferred the functional form of the structural unidentifiabilities, through analysis of the trajectories generated by the algorithm. In general, this inference may not be trivial. It is considerably simplified by considering the trajectories in log-space, as depicted in Figure \ref{fig:collatedSubspaces}, so long as the unidentifiability is rational. The reason is that if there is an unidentifiability over a relation of the form $p(\theta) = 0$, where $p(\theta)$ is rational, then a linear function of the logarithms of the parameters is preserved, and inference of the functional relation becomes a linear regression problem.

A key theme of this paper is that the parametric sensitivity characteristics of a mathematical model can depend highly on both the parameter vector considered and the length-scale of parameter perturbation. We saw in the NF-$\kappa$B network example that there were several practical unidentifiabilities that held in a certain region of parameter space, but had a distinct boundary at which a step change in the unidentifiability occurred (see e.g. Figure \ref{fig:collatedSubspaces}b, Figure \ref{fig:collatedSubspaces}c). These corresponded to the regions in which particular model approximations were valid. For instance, some such unidentifiabilities involved the rate constants of fast-timescale reactions, which were unidentifiable as long as a timescale separation existed between the reaction dynamics and observed model dynamics, but regained identifiability once they fell below a certain level. Our algorithm identified both the practical unidentifiabilities, and the boundaries in parameter space at which the form of the unidentifiabilities experienced a qualitative change. It is therefore of use in the model reduction problem, both for identifying redundant mechanisms, and quantifying the regions of parameter space at which they remain redundant. 

The majority of the computational time is taken up simulating the model at different points in parameter space, in order to calculate the associated cost function, and gradient. The density of simulations on a minimally disruptive curve varies with the `momentum' (i.e. the choice of magnitude of the initial costate $\lambda(0)$), the choice of ODE solver used to evolve the Hamiltonian flow (we used the `ode113' solver from the MATLAB toolbox \cite{MATLAB}), and the model in question. We took $C_{\theta^*}(\gamma(F)) = 0.1$ to set our momentum, for all curves generated in the previous example. $440$ internal evaluations of the cost function and its gradient were required by the ode113 solver to produce Figure \ref{fig:collatedSubspaces}a, while Figure \ref{fig:collatedSubspaces}e required $1363$ evaluations. Each combined evaluation of the cost function and gradient took about five seconds on a single-core, 3.40 GHz processor. The gradient was calculated by solving the sensitivity equations associated with the model (see \cite{Khalil2002}), rather than by finite-difference methods.

Consideration of model sloppiness at a non-local level has been considered previously for locally identifiable models (see \cite{Transtrum2010, TranstrumMS2011}). These approaches come from a different angle, and consideration of the varying purposes behind our formulation and theirs is warranted. In the approach of \cite{Transtrum2010}, the cost Hessian is taken as a Riemannian Metric on parameter space. The distance between parameter vectors is then taken as the length of a geodesic with respect to this metric. Given some nominal parameter vector $\theta^*$, geodesics with initial direction along the eigenvalues of the cost Hessian are evolved until the borders of parameter space are reached. The ratio of the longest and shortest geodesic lengths then determines the global sloppiness of the model. 

The global sloppiness quantification outlined has the advantage of being independent of the parameter units: geodesic lengths are invariant to the choice of co-ordinates. They do, however, depend on the choice of parameter $\theta^*$ from which they are initialised, which highly affects the anisotropy of the cost Hessian, as previously demonstrated for the oscillator model given in \eqref{eq:reparamOsc}. Evolution of the geodesic relies on successive linearisations of the cost function relative to the current location of the geodesic, rather than $\theta^*$, whose dynamics are `forgotten'. Thus it is determined by the infinitesimal sensitivity properties of the model, relative to a constantly changing nominal parameterisation. We instead consider non-infinitesimal sensitivity properties relative  to a fixed nominal parameterisation, or data, which allows for a link between multiscale sloppiness and the confidence regions of parameter estimates fitted to significantly noisy data. The distinction was discussed in Subsection \ref{ex:exp}, using a simple exponential decay model taken from \cite{Transtrum2010} (see also Figure \ref{fig:expComp}). Note that minimally disruptive curves and geodesics of the cost Hessian are relevant on different classes of model. Minimally disruptive curves are formulated specifically to recover structural and practical unidentifiabilities. Calculation of the acceleration term for the geodesic described, by contrast, involves inversion of the cost Hessian. This is by definition singular in locally structurally unidentifiable models, and numerically singular in locally practically unidentifiable models, and so cannot be inverted.

\section{Conclusion}
 In conclusion, we have presented a scalable algorithm allowing for systematic characterisation of the structural and practical unidentifiabilities in a model. The approach is conceptually novel, treating unidentifiable cross-sections of parameter space as solutions of an optimal control problem, which is then solved numerically.  En route, we have formulated a new approach to quantifying model sensitivity to parameter perturbations of non-infinitesimal magnitude. In the dual regime of estimation, we showed how this corresponds to the confidence regions of parameter estimates in the presence of non-infinitesimal measurement noise. The scalability, accuracy, and utility of the approach have been demonstrated through the discovery of new unidentifiabilities in benchmark models. Our hope is that once a model of a process is constructed, our methods can be used to systematically identify hidden unidentifiabilities and mechanistic redundancies. Subsequent parameter estimation and model-based analysis of the process is then made more tractable and reliable.

\appendix

\section{The correspondence between structural unidentifiability, sloppiness, and multiscale sloppiness} \label{app:section2}

In the main text we define a non-negative cost function $C_{\theta^*} \in \mathcal{C}^2$ on parameter space, such that $C_{\theta^*}(\theta^*) = 0$. Here $\mathcal{C}^k$ denotes the space of $k$-times differentiable, continuous functions. We then define multiscale sloppiness, for parameter perturbations of length scale $\delta$, and in the absence of data, as
\begin{align*}
& {\cS}_{\theta^*}(\delta) := \frac{{\cS}^{max}_{\theta^*}(\delta)}{{\cS}^{min}_{\theta^*}(\delta)} := \frac{\max_{\theta} C_{\theta^*}(\theta) : \|\theta - \theta^*\|^2_2 \leq \delta}{\min_\theta  C_{\theta^*}(\theta) : \|\theta - \theta^*\|^2_2 = \delta }  \\
& {\cS}_{\theta^*}(0) :=  \lim_{\delta \to 0} {\cS}_{\theta^*}(\delta).
\end{align*}

We state that ${\cS}_{\theta^*}(0)$ agrees with the traditional quantification of sloppiness at $\theta^*$, i.e. the ratio of the maximum to the minimum eigenvalue of $\nabla^2C_{\theta^*}(\theta^*)$. We now prove the claim. Since $\theta^*$ is a local minimum of $C_{\theta^*}$, we have both that $\nabla_\theta C_{\theta^*}(\theta^*) = 0$, and that $\nabla^2_\theta C_{\theta^*}(\theta^*)$ is positive semidefinite. If we consider a perturbation $\delta \theta$, the remainder theorem ensures
\begin{align}
C_{\theta^*}(\theta^* + \delta \theta) = \langle \delta\theta,\nabla^2C_{\theta^*}(\theta^*)\delta \theta \rangle + O(\|\delta \theta\|^3_2). \label{eq:taylorCost}
\end{align}
 Therefore, for $\|\delta\theta\|^2_2 \leq \delta $, we have 
 \begin{align}
 &\lim_{\delta \to 0}\frac{\cS^n_{\theta^*}(\delta)}{\cS^d_{\theta^*}(\delta)}
 = \lim_{\delta \to 0} \displaystyle\frac{\max_{\delta\theta} \langle \delta\theta,\nabla^2C_{\theta^*}(\theta^*)\delta \theta \rangle}{\min_{\delta\theta} \langle \delta\theta,\nabla^2C_{\theta^*}(\theta^*)\delta \theta \rangle} : \label{eq:limitHess}
 \end{align}

Clearly the numerator and denominator of \eqref{eq:limitHess} must respectively lie in the eigenspace of the maximal and minimal eigenvectors of  $\nabla^2C_{\theta^*}(\theta^*)$.
  
{  
 We state in the main text that the shape of the $\epsilon$-uncertainty region {$U_{\theta^*}(\epsilon):= \{\theta: C_{\theta^*}(\theta) \leq \epsilon\}$}, approximates a hyper-ellipse in the limit of decreasing $\epsilon$ for structurally identifiable models. Let us formalise this statement mathematically.

 \begin{lemma} \textbf{Limiting behaviour of $\mathbf{U_{\theta^*}(\epsilon)}$} \label{lem:Ulimit}
Given a 
structurally identifiable parameter $\theta^*$, and a compact parameter space $\Theta$, consider the volume-normalised version of $U_{\theta^*}(\epsilon)$ given by
\begin{align}
 \tilde{U}_{\theta^*}(\epsilon) = \left\{ \theta^* + \frac{\delta\theta}{\sqrt{2\epsilon}}: \theta^* +  \delta \theta \in U_{\theta^*}\left({\epsilon}\right) \right\}. \label{eq:normalisedU}
\end{align}
 Then 
\begin{align*}
\lim_{\epsilon \to 0} \tilde{U}_{\theta^*}(\epsilon) = \theta^* + \left\{ \delta \theta:  \langle \delta\theta,\nabla^2C_{\theta^*}(\theta^*)\delta \theta \rangle \leq 1 \right\}.
\end{align*}
\end{lemma}

\begin{proof}
Since $\theta^*$ is structurally identifiable, we know that $U_{\theta^*}(0) = \{\theta^*\}$. Our first step is to show, for sufficiently small $\epsilon$, that $U_{\theta^*}(\epsilon)$ is a connected set. If not, then there exists a sequence $\{\theta^{(i)}\}_{i=1}^{\infty}$ such that, for all $i \in \mathbb{N}$:
\begin{align*}
\exists \kappa > 0 : \| \theta^{(i)} - \theta^*\|_2^2 > \kappa \text{ and } C_{\theta^*}(\theta^{(i)}) < \frac{1}{i}.
\end{align*}
The Extreme Value Theorem then guarantees the existence of a limiting parameter vector $\theta^{(\infty)} \in \Theta$ satisfying
\begin{align*}
\exists \kappa > 0 : \| \theta^{(\infty)} - \theta^*\|_2^2 > \kappa \text{ and } C_{\theta^*}(\theta^{(\infty)}) = 0,
\end{align*}
contradicting structural identifiability. 

Since  $U_{\theta^*}(\epsilon)$ is connected for sufficiently small $\epsilon$, and $C_{\theta^*}$ is continuous, we know that
\begin{align*}
\lim_{\epsilon \to 0} \left[ \max_{\delta \theta \in \Theta} \| \delta \theta \|_2^2: C_{\theta^*}(\theta^* + \delta \theta) < \epsilon \right] = 0.
\end{align*}

Next, taking a Taylor expansion of \eqref{eq:epsilonUncertain}, and using the definition of $\tilde{U}_{\theta^*}$ from \eqref{eq:normalisedU},  we see that 
\begin{align*}
&\tilde{U}_{\theta^*}(\epsilon) = \left\{ \delta \theta:  \left( \frac{ (\sqrt{2 \epsilon})^2}{2} \langle \delta \theta, [\nabla^2_\theta C_{\theta^*} (\theta^*)] \delta \theta \rangle \right) + \mathcal{O}\left( \|2 \sqrt{\epsilon} \delta \theta\|_2^3 \right) \leq \epsilon \right\}. 
\end{align*}
We can divide the above equation through by $\epsilon$, and note, by definition, that
\begin{align*}
\lim_{\epsilon \to 0} \frac{\mathcal{O}\left( \|2 \sqrt{\epsilon} \delta \theta\|_2^3\right)}{\epsilon} = 0,
\end{align*}
to get
\begin{align*}
\lim_{\epsilon \to 0} \tilde{U}_{\theta^*}(\epsilon)  = \left\{ \delta \theta:  \left(  \langle \delta \theta, [\nabla^2_\theta C_{\theta^*} (\theta^*)] \delta \theta \rangle \right)  \leq 1 \right\}
\end{align*}
as required. 
\end{proof}

 }


 The main text states that structural unidentifiability is equivalent to the existence of a strictly positive $\delta$ such that $\cS_{\theta^*}(\delta)$ is infinite. We prove the forward implication first. Structural unidentifiability implies there exists $\hat{\theta} \neq \theta^*$ such that $C_{\theta^*}(\hat{\theta}) = 0$. So if we take $\delta = \|\hat{\theta} - \theta^*\|_2^2$, then we are assured that $\delta > 0$. But this means that $\cS^d_{\theta^*}(\delta) = 0$. We now prove the backward implication. Continuity of $C_{\theta^*}$ is assumed. Therefore $C_{\theta^*}$ must attain a finite maximum value on any compact set, by the Extreme Value Theorem. This implies that $\cS^n_{\theta^*}(\delta)$ is finite. So $\cS_{\theta^*}(\delta)$ can only be infinite,  if $\cS^d_{\theta^*}(\delta) = 0$. {Therefore} there exists {a} $\hat{\theta}$ such that $\|\hat{\theta} - \theta\|_2^2 = \delta$ and $C_{\theta^*}(\hat{\theta}) = 0$. Since $\delta > 0$, we are assured that $\hat{\theta} \neq \theta^*$, and thus $\hat{\theta}$ satisfies the conditions for structural unidentifiability.
 
  {
 \section{The correspondence between minimally disruptive parameters and minimally disruptive curves} \label{app:section4}

The Lemma below shows that continuous branches of minimally disruptive parameters are minimally disruptive curves.

 \begin{lemma} \label{lem:mdCurves}
 
Suppose a function $\Xi^\Delta_{min}(\delta)$ satisfies
 \begin{align}
& \Xi^\Delta_{min}(\delta) \in D^{min}_{\theta^*}(\delta) \ \forall \delta \in [0,\Delta], \label{appeq:choiceD} 
 \end{align}
 for some $\Delta > 0$, and where $D^{min}_{\theta^*}(\delta)$ denotes the minimally disruptive parameter at length-scale $\delta$. Then $\Xi^\Delta_{min}(\delta)$ is a solution of the following optimization problem:
   \begin{align}
& \gamma^* = \min_{\gamma \in \Gamma^F_{\theta^*}} : \int_{\gamma} C_{\bullet}: \label{appeq:ldMin} 
\\ \Gamma^F_{\theta^*} =  
&\left\{  \gamma:[0,F] \to \Theta: \begin{array}{ll} \gamma(0) = \theta^*;  &\frac{\mathrm{d}}{\mathrm{d}s} \|\gamma(s) - \theta^*\|_2 > 0 \\ \|\gamma'(s)\|_2=1
   \end{array}\right\}. \nonumber
\end{align}
\end{lemma}

\begin{proof} First note that for any $\gamma \in \Gamma^F_{\theta^*}$, we can take the reparameterisation 
\begin{align}
\delta(s) =  \|\gamma(s) - \theta^*\|_2. \label{appeq:deltaRep}
\end{align}
We have, for $s> 0$, that  $\frac{\mathrm{d}}{\mathrm{d}s} \|\gamma(s) - \theta^*\|_2 > 0$. This ensures that $\delta(s)$ is a bijective function, and so it has an inverse. We can therefore parameterise the line integral of \eqref{appeq:ldMin} using either $\delta$ or $s$. Note however that the reparameterisation $\delta(s)$ will depend on the specific curve $\gamma$. 

Suppose that a function $\Xi^\Delta_{min}$ satisfying \eqref{appeq:choiceD}  exists, but is not a minimiser of \eqref{appeq:ldMin}. We take the actual optimiser of \eqref{appeq:ldMin} as $\gamma^*$ . Reparameterisations of the form \eqref{appeq:deltaRep} will be different for $\Xi^\Delta_{min}$ and $\gamma^*$, so we denote them $\delta^1(s)$ and $\delta^2(s)$ respectively. Now by assumption, we have that
\begin{align}
& \int_{0}^{\delta^1(F)} C_{\bullet} \big(\gamma^*(\delta )\big) \left\| \frac{\mathrm{d}}{\mathrm{d} \delta} \big[\gamma^*(\delta) \big] \right\|_2 \ d \delta 
< &\int_{0}^{\delta^2(F)} C_{\bullet} \big(\Xi^\Delta_{min}(\delta )\big)  \left\| \frac{\mathrm{d}}{\mathrm{d} \delta} \left[ \Xi^\Delta_{min}(\delta )\right]  \right\|_2 \  d \delta. \label{eq:minDisIneq}
\end{align}

Since $\Xi^\Delta_{min}$, by definition, only traces over minimally disruptive parameters, we know, for any $\delta \in \big(0,\delta^2(F)\big]$, that
\begin{align*}
C_{\bullet} \big(\gamma^*(\delta )\big) \geq  C_{\bullet} \big(\Xi^\Delta_{min}(\delta )\big).
\end{align*}
Furthermore the inequality is strict for at least one such $\delta$, in order that $\gamma^*$ itself does not only trace over minimally disruptive parameters. For \eqref{eq:minDisIneq} to hold, we therefore require 
\begin{align*}
\int_{0}^{\delta^1(F)}   \left\| \frac{\mathrm{d}}{\mathrm{d} \delta} \big[\gamma^*(\delta) \big] \right\|_2  \ d \delta < \int_{0}^{\delta^2(F)}  \left\| \frac{\mathrm{d}}{\mathrm{d} \delta} \left[ \Xi^\Delta_{min}(\delta )\right]  \right\|_2 \  d \delta.
 \end{align*} 
The LHS and RHS of the above inequality are equal to the lengths of the curves $\gamma^*$ and $\Xi^\Delta_{min}$, respectively. So $\gamma^*$ must be strictly shorter than  $\Xi^\Delta_{min}$. However, this contradicts \eqref{appeq:ldMin}, which demands that both $\gamma^*$ and $\Xi^\Delta_{min}$ have length $F$. So $\gamma^*$ cannot exist. 
\end{proof}
 }

\section{Necessary conditions on minimally disruptive curves} \label{app:section3}
 \subsection{The solution of an abstract optimal control problem}
We begin by providing the outline of a standard solution to an abstract optimal control problem. This problem is later specialised to provide necessary conditions on the minimally disruptive curves described in the paper. A more detailed solution of the abstract problem is provided in e.g. \cite[Ch.3]{Bryson1975}. 

Consider a controlled dynamical system of the following form:
\begin{align*}
&\dot{x}(t) = f(x(t),u(t)) &x(0) = x_0 \\
&x(t) \in \mathbb{R}^n &u(t) \in \mathbb{R}^m
\end{align*}
Here $t \in \mathbb{R}^+$ denotes time, while $u(t)$ is a vector of time-varying control inputs, which the user may vary so as to affect dynamics of {the state}, $x(t)$. Dotted variables denote time derivatives. Given a globally nonegative potential $L(x,u) \in \mathbb{R}^+$, we wish to find necessary conditions on $u(t)$ for minimisation of the following action:
\begin{subequations}
\begin{align}
J(u) = \int^{t_F}_0 L(x(t),u(t))\ dt, \label{eq:abstractFunctional}
\end{align}
for some fixed time $t_F$, and subject to constraints of the form
\begin{align}
&K_1(x(t),u(t)) \leq 0, &K_2(x(t),u(t) = 0. \label{eq:inputConstraints}
\end{align}
\label{eq:optConds}
\end{subequations}
Note that along trajectories of the dynamical system, \eqref{eq:abstractFunctional} can be rewritten as
\begin{align*}
J(u) =  \int^{t_F}_0 L(x(t),u(t)) + \langle \lambda(t), \big[ \dot{x}(t) - f(x(t),u(t))   \big] \rangle \ dt,
\end{align*}
for any function $\lambda: [0,t_F] \to \mathbb{R}^n$. 
{Consider the } Hamiltonian:
\begin{align}
H(x,u,\lambda) = L(x,u) + \langle\lambda(t),f(x,u)\rangle. \label{eq:vanillaHam}
\end{align}
{Integrating by parts with respect to $t$ gives}
\begin{align*}
J(u) = &\int^{t_F}_0 H(x(t),u(t),\lambda(t)) + \langle \dot{\lambda}(t),x(t)\rangle \ dt \\ + &\langle\lambda(0),x_0\rangle - \langle\lambda(t_F),x(t_F)\rangle.
\end{align*}
From this, the first variation of $J$, given fixed $x_0$ and free $x(t_F)$, is given by
\begin{align*}
\begin{split}
\delta J(\delta u) &=\int^{t_F}_0 \left \langle \left[ \frac{\partial H}{\partial x}\big(x(t),u(t),\lambda(t) \big) + \dot{\lambda}(t)\right], \delta x(t) \right\rangle  \ dt  \\
& + \int^{t_F}_0  \left\langle \frac{\partial H}{\partial u}\big(x(t),u(t),\lambda(t) \big), \delta u(t) \right\rangle\ dt  - \lambda(t_F)\delta x(t_F).
\end{split}
\end{align*}
Since $\lambda(t)$, which we term the costate, has thus far been arbitrary, we {impose}
\begin{align*}
&\dot{\lambda}(t) = -\frac{\partial H}{\partial x}\big(x(t),u(t),\lambda(t)\big) 
&\lambda(t_F) = 0,
\end{align*}
to get
\begin{align*}
\delta J(\delta u) = &\int^{t_F}_0 \left \langle \frac{\partial H}{\partial u}\big(x(t),u(t),\lambda(t)\big), \delta u(t) \right\rangle \ dt.
\end{align*}
The form of the costate derivative also ensures time-invariance of the Hamiltonian, as can be seen by explicit differentiation.

Without the presence of the constraints \eqref{eq:inputConstraints}, the condition $\delta J = 0$ for all perturbations $\delta u:[0,t_F] \to \mathbb{R}^m$, would be a necessary condition for first order optimality. This is equivalent to the condition $ \frac{\partial H}{\partial u}\big(x(t),u(t),\lambda(t)\big) = 0$ for all $t \in [0,t_F]$. As we are in fact looking for a constrained minimum, we only need deal with perturbations $\delta u$ preserving the constraints. For this, we repeat the previous arguments, but apply the Karush-Kuhn-Tucker (KKT) conditions on constrained optima at each timepoint. This gives a modified Hamiltonian:
\begin{align}
\tilde{H}(x,u,\lambda) = H(x,u,\lambda) + \mu_1(t) K_1(x,u) + \mu_2(t)K_2(x,u), \label{eq:modHam}
\end{align}
where $\mu_1(t)$ is a time-varying KKT multiplier satisfying
\begin{align*}
&\mu_1(t) \geq 0 & \mu_1(t) K_1\big(x(t),u(t)\big)= 0,
\end{align*}
and $\mu_2(t)$ is a time varying Lagrange multiplier. The necessary conditions on optimality become:
\begin{subequations}
\begin{align}
&\dot{\lambda}(t) = -\frac{\partial \tilde{H}}{\partial x}\big(x(t),u(t),\lambda(t)\big) \label{eq:costateDeriv}\\
&\lambda(t_F) = 0 \label{eq:finalCostate} \\ 
&\frac{\partial \tilde{H}}{\partial u}(x(t),u(t),\lambda(t)) = 0 \ \ \ \forall t \in [0,t_F]. \label{eq:stationarityCond}
\end{align}
\end{subequations}
Note that the value of $\tilde{H} - H$ is null by definition, although their partial derivatives differ. The value of $\tilde{H}$ itself is set by \eqref{eq:finalCostate} and \eqref{eq:vanillaHam}, which together imply that 
\begin{align}
\tilde{H}(x(t),u(t),\lambda(t)) = L(x(t_F),u(t_F)). \label{eq:HamValue}
\end{align}

\subsection{Specialisation to minimally disruptive curves} \label{app:min}
Consider a model with nominal parameter vector $\theta^*$, and a non-negative cost function $C_{\bullet}(\theta)$, where $C_{\bullet}$ can represent either $C_{\theta^*}$ or $C_D$, depending on the presence of observational data. In the paper, a minimally disruptive curve $\gamma^*$ over $[0,\Delta]$ and relative to $\theta^*$, is defined as a solution of the following optimization problem:
\begin{align*}
\gamma^* = \min_{\gamma} \int^{F}_0 C_{\bullet}(\gamma(s)) \| \gamma'(s)\|_2 \ ds
\end{align*}
subject to
\begin{align*}
& \gamma(0) = \theta^* &\langle \gamma'(s),\gamma(s) - \theta^* \rangle > 0 \\
&\gamma \in \mathcal{C}^1 &\|\gamma'(s)\|_2 = 1.
\end{align*}
This is a special case of the optimization \eqref{eq:optConds}. In particular, the set of system states $x \in \mathbb{R}^n$ in \eqref{eq:optConds} is replaced by the set of parameter vectors $\theta \in \mathbb{R}^q$, and the time dependent state trajectory $x(t)$ is replaced by the length-scale dependent curve $\gamma(s)$. We take complete control over $\gamma'(s)$ by setting $\gamma'(s) = u$. 
The Lagrangian forming the integrand of \eqref{eq:abstractFunctional} is then given by $C_{\bullet}(\gamma(s))\|u\|_2$. The constraints $K_1$ and $K_2$ are given by
\begin{subequations}
\begin{align}
&K_1\big (\gamma(s),u(s) \big) = - \langle u(s), \gamma(s) - \theta^*\rangle \\ &K_2\big (\gamma(s),u(s) \big) = \|u\|_2 - 1.
\end{align}
\end{subequations}
and the modified Hamiltonian \eqref{eq:modHam} is
\begin{align*}
&\tilde{H}(\gamma(s),u(s),\lambda(s)) = C_{\bullet}\big(\gamma(s)\big) +   \langle\lambda(s),u(s) \rangle \\& - \mu_1(s)\langle< u(s), \gamma(s) - \theta^*\rangle>  + \mu_2(s) \big(\|u\|_2 - 1\big).
\end{align*}

Thus the optimality condition \eqref{eq:stationarityCond} becomes:
\begin{align}
2\mu_2(s)u(s) - \mu_1(s)\big(\gamma(s)-\theta^*\big) + \lambda(s) = 0, \label{eq:ldStat}
\end{align}
while  \eqref{eq:HamValue} gives 
\begin{align*}
\tilde{H}\big(\gamma(s),u(s),\lambda(s)\big) &= C_{\bullet}\big(\gamma(F) \big)  
\end{align*}
If we take the inner product of \eqref{eq:ldStat} with $u(s)$, while heeding the constaints, we see that
\begin{align*}
&2\mu_2(s) = C_{\bullet}(\gamma(s)) - \tilde{H}\big(\gamma(s),u(s),\lambda(s)\big).
\end{align*}
Given $\gamma(s)$, $\lambda(s)$ and $C_{\theta^*}\big(\gamma(s)\big)$, both $\mu_1(s)$ and $u(s)$ can now found as the solution of a {system of linear equations}. Given a user-specified $u(0)$ and $C_{\theta^*}(\gamma(F))$, the pair of curves $\{\gamma(s),\lambda(s)\}$ can be explicitly evolved as an ordinary differential equation (ODE). Note that $F$ need not be defined explicitly.

\begin{figure}[h]                                 
\includegraphics[width=12cm]{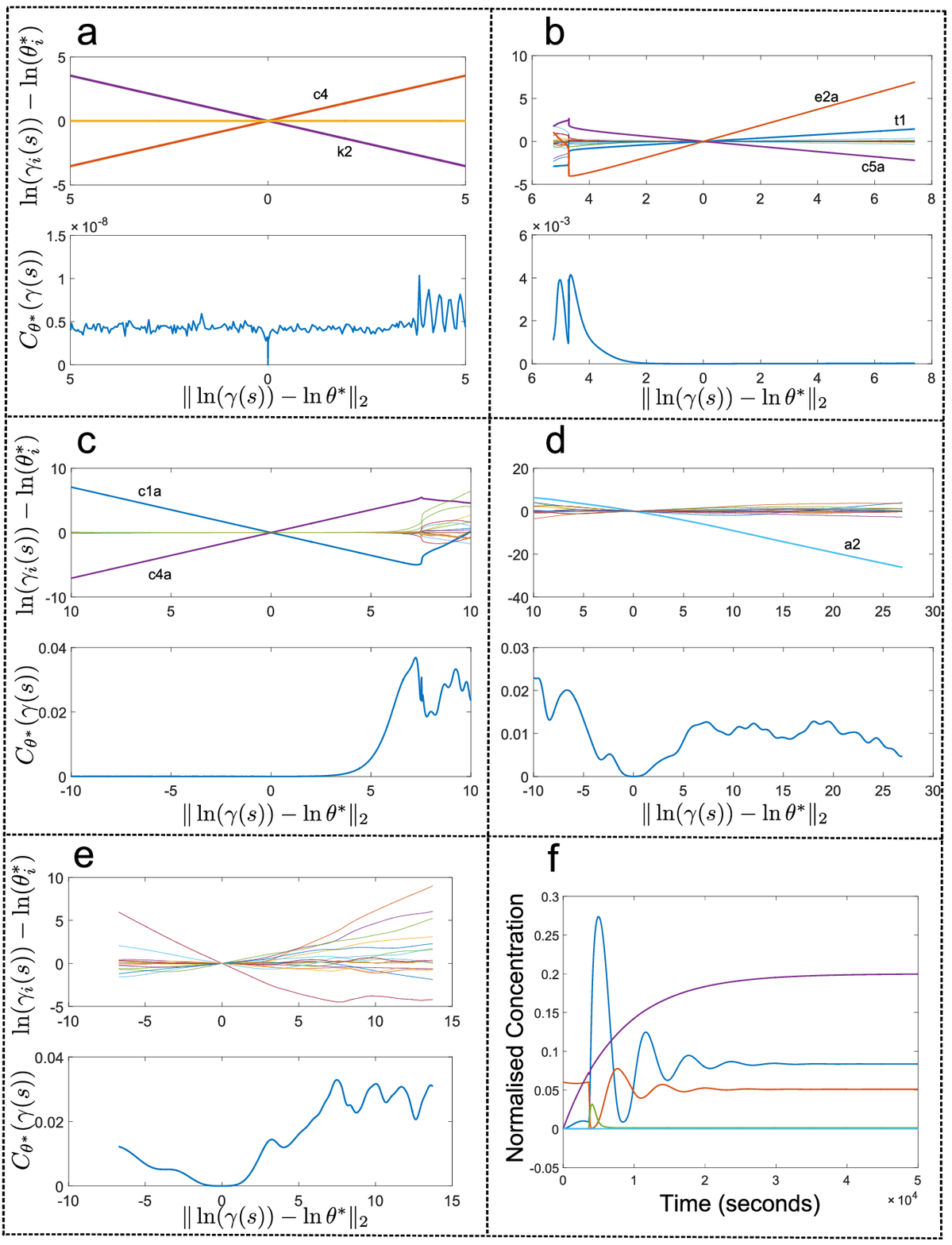}
\flushleft
\caption{\label{fig:collatedSubspaces} Panels a-e depict structural and practical unidentifiabilities in the NF-$\kappa$B model from \cite{Lipniacki2004}. The $x$-axes of all plots in these panels depict the distance, in log-space, between the nominal parameter vector $\theta^*$, and that of $\gamma(s)$, as $s$ varies. Each line on the top-half figures of these panels represents the relative change of a particular parameter component $\theta_i$, from $\theta^*_i$, as $s$ varies. The bottom half figures represent the cost associated with the parameter vector $\gamma(s)$, relative to $\theta^*$, as $s$ varies. Panel f provides the normalised time-course of the observed quantities in the model. }
 \end{figure}

\bibliographystyle{plain}
\bibliography{ar2Bib}

\begin{thebibliography}{10}

\bibitem{Anderson2009}
J.~Anderson and A.~Papachristodoulou.
\newblock {On validation and invalidation of biological models}.
\newblock {\em BMC bioinformatics}, 10(1):132, 2009.

\bibitem{Anguelova2007}
M.~Anguelova.
\newblock {\em {Observability and identifiability of nonlinear systems with
  applications in biology}}.
\newblock PhD thesis, Chalmers University of Technology, 2007.

\bibitem{Apgar2010}
J.~F. Apgar, D.~K. Witmer, F.~M. White, and B.~Tidor.
\newblock {Sloppy models, parameter uncertainty, and the role of experimental
  design.}
\newblock {\em Molecular bioSystems}, 6(10):1890--900, 2010.

\bibitem{Bellman1970}
R.~Bellman and \AA{str\"{o}m, K.J.}
\newblock {On structural identifiability}.
\newblock {\em Mathematical Biosciences}, 7(3):329--339, 1970.

\bibitem{Brenner2010}
S.~Brenner.
\newblock Sequences and consequences.
\newblock {\em Philosophical Transactions of the Royal Society B: Biological
  Sciences}, 365(1537):207--212, 2010.

\bibitem{Brown2003}
K.~S. Brown and J.~P. Sethna.
\newblock {Statistical mechanical approaches to models with many poorly known
  parameters}.
\newblock {\em Physical Review E}, 68(2):021904, 2003.

\bibitem{Bryson1975}
A.E. Bryson and Ho~Y-H.
\newblock {\em {Applied optimal control: optimization, estimation and
  control}}.
\newblock Taylor {\&} Francis, 1975.

\bibitem{Casella2002}
G.~Casella and R.L. Berger.
\newblock {\em {Statistical inference}}.
\newblock Duxbury, 2nd edition, 2002.

\bibitem{Chis2011}
O-T. Chis, J.R. Banga, and E.~Balsa-Canto.
\newblock {Structural identifiability of systems biology models: a critical
  comparison of methods}.
\newblock {\em PloS one}, 6(11):e27755, 2011.

\bibitem{Chis2014}
O-T. Chis, J.R. Banga, and E.~Balsa-Canto.
\newblock {Sloppy models can be identifiable}.
\newblock {\em arXiv preprint}, (arXiv:1403.1417), 2014.

\bibitem{Denis-Vidal2000}
L.~Denis-Vidal and G.~Joly-Blanchard.
\newblock {An easy to check criterion for (un)indentifiability of uncontrolled
  systems and its applications}.
\newblock {\em IEEE Transactions on Automatic Control}, 45(4):768--771, 2000.

\bibitem{Denis-Vidal2004}
L.~Denis-Vidal and G.~Joly-Blanchard.
\newblock {Equivalence and identifiability analysis of uncontrolled nonlinear
  dynamical systems}.
\newblock {\em Automatica}, 40(2):287--292, 2004.

\bibitem{Evans2002}
N.D. Evans, M.J. Chappell, and K.R. Godfrey.
\newblock {Identifiability of uncontrolled nonlinear rational systems}.
\newblock {\em Automatica}, 38(10):1799--1805, 2002.

\bibitem{Gutenkunst2007}
R.~N. Gutenkunst, J.~J. Waterfall, F.~P. Casey, K.~S. Brown, C.~R. Myers, and
  J.~P. Sethna.
\newblock {Universally sloppy parameter sensitivities in systems biology
  models}.
\newblock {\em PLoS computational biology}, 3(10):1871--78, 2007.

\bibitem{Hengl2007}
S.~Hengl, C.~Kreutz, J.~Timmer, and T.~Maiwald.
\newblock {Data-based identifiability analysis of non-linear dynamical models.}
\newblock {\em Bioinformatics}, 23(19):2612--8, 2007.

\bibitem{Hermann1977}
R.~Hermann and A.~J. Krener.
\newblock {Nonlinear controllability and observability}.
\newblock {\em IEEE Transactions on automatic control}, 22(5):728--740, 1977.

\bibitem{Hines2014}
K.E. Hines, T.~R. Middendorf, and R.~W. Aldrich.
\newblock {Determination of parameter identifiability in nonlinear biophysical
  models: A Bayesian approach}.
\newblock {\em The Journal of General Physiology}, 143(3):401--16, 2014.

\bibitem{Joshi2006}
M.~Joshi, A.~Seidel-Morgenstern, and A.~Kremling.
\newblock {Exploiting the bootstrap method for quantifying parameter confidence
  intervals in dynamical systems}.
\newblock {\em Metabolic engineering}, 8(5):447--55, 2006.

\bibitem{Khalil2002}
H~Khalil.
\newblock {\em {Nonlinear Systems Third Edition}}.
\newblock Prentice Hall Upper Saddle River, 3rd edition, 2002.

\bibitem{Letellier2005}
C.~Letellier, L.~A. Aguirre, and J.~Maquet.
\newblock {Relation between observability and differential embeddings for
  nonlinear dynamics}.
\newblock {\em Physical Review E}, 71(6):066213, 2005.

\bibitem{Lipniacki2004}
T.~Lipniacki, P.~Paszek, and A.R. Brasier.
\newblock {Mathematical model of NF-$\kappa$B regulatory module}.
\newblock {\em Journal of theoretical biology}, pages 195--215, 2004.

\bibitem{Liu2013}
Y-Y. Liu, J-J. Slotine, and A-L. Barab{\'{a}}si.
\newblock {Observability of complex systems.}
\newblock {\em Proceedings of the National Academy of Sciences},
  110(7):2460--5, 2013.

\bibitem{Ljung1994}
L.~Ljung and T.~Glad.
\newblock {On global identifiability for arbitrary model parametrizations}.
\newblock {\em Automatica}, 30(2):265--276, 1994.

\bibitem{Margaria2001}
G.~Margaria, E.~Riccomagno, M.~J. Chappell, and H.~P. Wynn.
\newblock {Differential algebra methods for the study of the structural
  identifiability of rational function state-space models in the biosciences}.
\newblock {\em Mathematical Biosciences}, 174(1):1--26, 2001.

\bibitem{MATLAB}
MathWorks.
\newblock Matlab and statistics toolbox release 2013b, 2013.

\bibitem{Meshkat2009}
N.~Meshkat, M.~Eisenberg, and J.~J. Distefano.
\newblock {An algorithm for finding globally identifiable parameter
  combinations of nonlinear ODE models using Gr{\"{o}}bner Bases.}
\newblock {\em Mathematical biosciences}, 222(2):61--72, 2009.

\bibitem{Moles2003}
C.~G. Moles, P.~Mendes, and J.~R. Banga.
\newblock Parameter estimation in biochemical pathways: a comparison of global
  optimization methods.
\newblock {\em Genome research}, 13(11):2467--2474, 2003.

\bibitem{Peixoto2014}
T.~P. Peixoto.
\newblock {Model Selection and Hypothesis Testing for Large-Scale Network
  Models with Overlapping Groups}.
\newblock {\em Physical Review X}, 4(1):011033, 2014.

\bibitem{Pohjanpalo1978}
H~Pohjanpalo.
\newblock {System identifiability based on the power series expansion of the
  solution}.
\newblock {\em Mathematical biosciences}, 1978.

\bibitem{Raia2011}
V.~Raia, M.~Schilling, M.~B{\"{o}}hm, B.~Hahn, A.~Kowarsch, A.~Raue, C.~Sticht,
  S.~Bohl, M.~Saile, P.~M{\"{o}}ller, N.~Gretz, J.~Timmer, F.~Theis, W.-D.
  Lehmann, P.~Lichter, and U.~Klingm{\"{u}}ller.
\newblock {Dynamic mathematical modeling of IL13-induced signaling in Hodgkin
  and primary mediastinal B-cell lymphoma allows prediction of therapeutic
  targets}.
\newblock {\em Cancer research}, 71(3):693--704, 2011.

\bibitem{Raman2016}
D.V. Raman, J.~Anderson, and A.~Papachristodoulou.
\newblock {On the performance of nonlinear dynamical systems under parameter
  perturbation}.
\newblock {\em Automatica}, 63:265--273, 2016.

\bibitem{Raue2014}
A.~Raue, J.~Karlsson, and M.P. Saccomani.
\newblock {Comparison of approaches for parameter identifiability analysis of
  biological systems}.
\newblock {\em Bioinformatics}, 2014.

\bibitem{Raue2009}
A.~Raue, C.~Kreutz, T.~Maiwald, J.~Bachmann, M.~Schilling, U.~Klingm\"{u}ller,
  and J.~Timmer.
\newblock {Structural and practical identifiability analysis of partially
  observed dynamical models by exploiting the profile likelihood}.
\newblock {\em Bioinformatics}, 25(15):1923--9, 2009.

\bibitem{Sedoglavic2008}
A.~Sedoglavic.
\newblock {A probabilistic algorithm to test local algebraic observability in
  polynomial time}.
\newblock In {\em Proceedings of the International Symposium on Symbolic and
  Algebraic Computation}. ACM, 2001.

\bibitem{TranstrumMS2011}
M.~K. Transtrum, B.~B. Machta, and J.P. Sethna.
\newblock {Geometry of nonlinear least squares with applications to sloppy
  models and optimization}.
\newblock {\em Physical Review E}, 83(3):036701, 2011.

\bibitem{Transtrum2014}
M.~K. Transtrum and P.~Qiu.
\newblock {Model reduction by manifold boundaries}.
\newblock {\em Physical Review Letters}, 113(9):098701, 2014.

\bibitem{Transtrum2010}
M.K. Transtrum, B.~B. Machta, and J.~P. Sethna.
\newblock {Why are nonlinear fits to data so challenging?}
\newblock {\em Physical Review Letters}, 104(6):060201, 2010.

\bibitem{Vajda1989b}
S.~Vajda, K.~R. Godfrey, and H.~Rabitz.
\newblock {Similarity transformation approach to identifiability analysis of
  nonlinear compartmental models}.
\newblock {\em Mathematical Biosciences}, 93(2):217--248, 1989.

\bibitem{Vajda1989}
S.~Vajda, H.~Rabitz, E.~Walter, and Y.~Lecourtier.
\newblock {Qualitative and quantitative identifiability analysis of nonlinear
  chemical kinetic models}.
\newblock {\em Chemical Engineering Communications}, 83(1):191--219, 1989.

\bibitem{Vallisneri2008}
M.~Vallisneri.
\newblock {Use and abuse of the Fisher information matrix in the assessment of
  gravitational-wave parameter-estimation prospects}.
\newblock {\em Physical Review D}, 77(4):042001, 2008.

\bibitem{Waterfall2006}
J.~J. Waterfall, F.~P. Casey, R.~N. Gutenkunst, K.~S. Brown, C.~R. Myers, P.~W.
  Brouwer, V.~Elser, and J.~P. Sethna.
\newblock {The sloppy model universality class and the Vandermonde matrix}.
\newblock {\em Physical Review Letters}, 97(15):150601, 2006.

\bibitem{Whalen2015}
A.~J. Whalen, S.~N. Brennan, T.~D. Sauer, and S.~J. Schiff.
\newblock {Observability and Controllability of Nonlinear Networks: The Role of
  Symmetry}.
\newblock {\em Physical Review X}, 5(1):011005, 2015.

\bibitem{Xia2003}
X.~Xia and C.~H. Moog.
\newblock {Identifiability of nonlinear systems with application to HIV/AIDS
  models}.
\newblock {\em IEEE Transactions on Automatic Control}, 48(2):330--336, 2003.

\bibitem{Yates2009}
J.~W.T. Yates, N.D. Evans, and M.J. Chappell.
\newblock {Structural identifiability analysis via symmetries of differential
  equations}.
\newblock {\em Automatica}, 45(11):2585--2591, 2009.

\end{thebibliography}
\end{document}